%% file: vertex-interdiction-full.tex
\documentclass[runningheads]{llncs}

\usepackage{amssymb,amsmath,multicol,enumerate,amsthm}

\usepackage{blindtext}
\usepackage[utf8]{inputenc}

\usepackage{latexsym}
\usepackage{algorithm}
\usepackage{algorithmic}
\usepackage{graphicx}
\usepackage{blindtext}

\begin{document}
	
\title{Vertex Downgrading to Minimize Connectivity \thanks{This material is based upon research supported in part by the U. S. Office of Naval Research under award number N00014-18-1-2099}
}
\institute{}
	
	\author{\textsc{
			Hassene Aissi\thanks{Paris Dauphine University.  {\tt aissi@lamsade.dauphine.edu}},
			Da Qi Chen\thanks{Carnegie Mellon University.  {\tt daqic@andrew.cmu.edu}},
			R. Ravi\thanks{Tepper School of Business, Carnegie Mellon University. {\tt ravi@cmu.edu}}}}
	
	\maketitle
	
\begin{abstract}
We study the problem of interdicting a directed graph by deleting nodes with the goal of minimizing the local edge connectivity of the remaining graph from a given source to a sink.
We show hardness of obtaining strictly unicriterion approximations for this basic vertex interdiction problem.
We introduce and study a general downgrading variant of the interdiction problem where the capacity of an arc is a function of the subset of its endpoints that are downgraded, and the goal is to minimize the downgraded capacity of a minimum source-sink cut subject to a node downgrading budget.
This models the case when both ends of an arc must be downgraded to remove it, for example.
For this generalization, we provide a bicriteria $(4,4)$-approximation that downgrades nodes with total weight at most 4 times the budget and provides a solution where the downgraded connectivity from the source to the sink is at most 4 times that in an optimal solution. WE accomplish this with an LP relaxation and round using a ball-growing algorithm based on the LP values.
We further generalize the downgrading problem to one where each vertex can be downgraded to one of $k$ levels, and the arc capacities are functions of the pairs of levels to which its ends are downgraded. We generalize our LP rounding to get $(4k,4k)$-approximation for this case.
\end{abstract}
	

\input{intro}
\input{nvdp}

\bibliographystyle{abbrv}
\bibliography{biblio2}

\newpage
\appendix
    \input{app-ipco}

    \newpage
	\input{app-nvdp}
	\input{nvldp}
	\input{hardness}

	\input{vip}

\end{document}

%% file: intro.tex
\section{Introduction}

Interdiction problems arise in evaluating the robustness of infrastructure and networks. For an optimization problem on a graph, the interdiction problem can be formulated as a game consisting of two players: an attacker and a defender. Every edge/vertex of the graph has an associated interdiction cost and the attacker interdicts the network by modifying the edges/vertices subject to a budget constraint. The defender solves the problem on the modified graph. The goal of the attacker is to hamper the defender as much as possible.
Ford and Fulkerson initiated the study of interdiction problems with the maximum flow/minimum cut theorem~\cite{chestnut2017hardness,harris55,schrijver2002history}.
Other examples of interdiction objectives include matchings~\cite{Zen10}, minimum spanning trees~\cite{LS17,Zen15}, shortest paths~\cite{golden1978problem,israeli2002shortest}, $st$-flows~\cite{Phillips:1993,Wood93,Zen10DAM} and global minimum cuts~\cite{Zen14,chestnut2016interdicting}.

Most of the interdiction literature today involves the interdiction of edges while the study of interdicting vertices has received less attention (e.g.\cite{Zen10,Zen10DAM}).
The various applications for these interdiction problems, including drug interdiction, hospital infection control, and protecting electrical grids or other military installations against terrorist attacks, all naturally motivate the study of the vertex interdiction variant.
In this paper, we focus on vertex interdiction problems related to the minimum $st$-cut (which is equal to the maximum $st$-flow and hence also termed network flow interdiction or network interdiction in the literature).


For $st$-cut vertex interdiction problems, the set up is as follows. Consider a directed graph $G=(V(G),A(G))$ with $n$ vertices and $m$ arcs, an arc cost function $c:A(G) \rightarrow \mathbb{N}$, and an interdiction cost function $r:V(G) \setminus \{s,t\} \rightarrow \mathbb{N}$ defined on the set of vertices $V(G) \setminus \{s,t\} $. A set of arcs $F\subseteq A(G)$ is an $st$-cut if $G\backslash F$ no longer contains a directed path from $s$ to $t$. Define the cost of $F$ as $c(F)= \Sigma_{e\in F} c(e)$.
For any subset of vertices $X \subseteq V(G) \setminus\{s,t\}$, we denote its interdiction cost by $r(X) = \sum_{v \in X} r(v)$. Let $\lambda_{st}(G\backslash X)$ denote the cost of a minimum $st$ cut in the graph $G\backslash X$.

\begin{problem}
	\textbf{Weighted Network Vertex Interdiction Problem (WNVIP) and its special cases.} Given two specific vertices $s$ (source) and $t$ (sink) in $V(G)$ and interdiction budget $b \in \mathbb{N}$, the Weighted Network Vertex Interdiction Problem {\bf (WNVIP)} asks to find an interdicting vertex set $X^* \subseteq V(G)\setminus \{s,t\}$ such that $\sum_{v \in X^*} r(v)\leq b$ and $\lambda_{st}(G\backslash X^*)$ is minimum. The special case of WNVIP where all the interdiction costs are unit will be termed {\bf NVIP}, while the further special case when even the arc costs are unit will be termed {\bf NVIP with unit costs}.

\end{problem}

In this paper, we define and study a generalization of the network flow interdiction problem in digraphs called vertex downgrading.
Since interdicting vertices can be viewed as attacking a network at a vertex, it is natural to consider a variant where attacking a node does not destroy it completely but partially weakens its structural integrity.
In terms of minimum $st$-cuts, one interpretation could be whenever a vertex is interdicted, instead of removing it from the network we partially reduce the cost of its incident arcs. In this context, we say that a vertex is \textit{downgraded}.
Specifically, consider a directed graph $G=(V(G), A(G))$ and a downgrading cost $r:V(G) \setminus \{s,t\} \rightarrow \mathbb{N}$. For every arc $e=uv\in A(G)$, there exist four associated nonegative costs $c_e, c_{eu}, c_{ev}, c_{euv}$, respectively representing the cost of arc $e$ if 1) neither $\{u, v\}$ are downgraded, 2) only $u$ is downgraded, 3) only $v$ is downgraded, and 4) both $\{u, v\}$ are downgraded. Note that these cost functions are independent of each other so downgrading vertex $v$ might affect each of its incident arcs differently. However, we do impose the following conditions: $c_e\ge c_{eu}\ge c_{euv}$ and  $c_e\ge c_{ev}\ge c_{euv}$. These inequalities are not restrictive and are natural to impose since the more endpoints of an arc are downgraded, the lower the resulting arc should cost. Given a downgrading set $Y\subseteq V(G) \setminus\{s,t\}$, define $c^Y:A(G)\to \mathbb{R}_+$ to be the arc cost function representing  the cost of cutting $e$ after downgrading $Y$.

\begin{center}
\begin{tabular}{|c||c|c|c|c|}
	\hline
	& $u, v\notin Y$ & $u\in Y, v\notin Y$ & $u\notin Y, v\in Y$ & $u, v\in Y$\\ \hline
	$c^Y(e)= $ & $c_e$ & $c_{eu}$ & $c_{ev}$ & $c_{euv}$\\ \hline
\end{tabular}
\end{center}

Given a set of arcs $F\subseteq A(G)$, we define $c^Y(F)=\Sigma_{e\in F} c^Y(e)$.


\begin{problem}
	\textbf{Network Vertex Downgrading Problem (NVDP).} Let $G=(V(G), A(G))$ be a directed graph with a source $s$ and a sink $t$. For every arc $e=uv$, we are given  non-negative costs $c_e, c_{eu}, c_{ev}, c_{euv}$ as defined above. Given a (downgrading) budget $b$, find a set $Y\subset V(G) \setminus\{s,t\}$ and an $st$-cut $F\subseteq A(G)$ such that $\Sigma_{v\in Y} r(v)\le b$ and minimizes $c^Y(F)$.
\end{problem}

While it is not immediately obvious as it is for WNVIP,  we can still show that detecting a zero solution for NVDP is easy. The proof of the following theorem is in Appendix \ref{app:nvdp}.

\begin{theorem}
	Given an instance of NVDP on graph $G$ with budget $b$, there exists a polynomial time algorithm to determine if there exists $Y\subseteq V(G)$ and an $st$-cut $F\subseteq A(G)$ such that $\Sigma_{v\in Y} r(v)\le b$ and $c^Y(F)=0$.
	\label{NVDPdetect}
\end{theorem}

First we present some useful reductions between the above  problems.

\begin{enumerate}
\item  In the NFI (Network Flow Interdiction) problem defined in \cite{chestnut2017hardness}, the given graph is undirected instead of directed and the adversary interdicts edges instead of vertices. The goal is to minimize the cost of the minimum $st$-cut after interdiction. NFI can be reduced to the undirected version of WNVIP (where the underlying graph is undirected). Simply subdivide every undirected edge $e=uv$ with a vertex $v_e$. The interdiction cost of $v_e$ remains the same as the interdiction cost of $e$ while all original vertices have an interdiction cost of $\infty$ (or a very large number). The cut cost of the edges $uv_e, v_ev$ are equal to the original cost of cutting the edge $e$.
\item The undirected version of WNVIP can be reduced to the (directed) WNVIP by replacing every edge with two parallel arcs going in opposite directions. Each new arc has the same cut cost as the original edge.
\item WNVIP is a special case of NVDP with costs $c_{eu}=c_{ev}=c_{euv}=0$ for all $e=uv$.
\end{enumerate}
The first two observations above imply that any hardness result for NFI in \cite{chestnut2017hardness} also applies to WNVIP.
Based on the second observation, we prove our hardness results for the (more specific) undirected version of WNVIP. As a consequence of the third observation, all of these hardness results also carry over to the more general NVDP.

Our work also studies the following further generalization of NVDP. Every vertex has $k$ possible levels that it can be downgraded to by paying different downgrading costs. Every arc has a cutting cost depending on what level its endpoints were downgraded to. More precisely, for each level $0\le i, j\le k$, let $r_i(v)$ be the interdiction cost to downgrade $v$ to level $i$ and let $c_{i, j}(e)$ be the cost of cutting arc $e=uv$ if $u, v$ were downgraded to levels $i, j$ respectively. We assume that $0=r_0(v)\le r_1(v)\le ...\le r_k(v)$ since higher levels of downgrading should cost more and $c_{i, j}(e)\ge c_{i', j'}(e)$ if $i\le i', j\le j'$ since the more one downgrades, the easier it is to cut the incident arcs. Then, given a map $L:V(G)\to \{0, ..., k\}$, representing which level to downgrade each vertex to, one can talk about the cost of performing this downgrading: $r^L:=\Sigma_{v\in V(G)} r_{L(v)}(v)$, and the cost of a cut $F$ after downgrading according to $L$: $c^L(F):= \Sigma_{uv\in F} c_{L(u), L(v)}(uv)$. Now, we can formally define the most general problem we address.
\begin{problem}
	\textbf{Network Vertex Leveling Downgrading Problem (NVLDP).} Let $G=(V(G), A(G))$ be a directed graph with a source $s$ and a sink $t$. For every vertex $v$ and $0\le i \le k$, we have non-negative downgrading costs $r_i(v)$. For every arc $e=uv$ and levels $0\le i, j \le k$, we are given  non-negative cut costs $c_{i, j}(e)$. Given a (downgrading) budget $b$, find a map $L:V(G)\to \{0, ..., k\}$ and an $st$-cut $F\subseteq A(G)$ such that $r^L\le b$ and minimizes $c^L(F)$.
\end{problem}

Note that when $k=1$ we have NVDP.





\subsection*{Related Works.}

\begin{definition}
An \textbf{$(\alpha, \beta)$ bicriteria approximation} for the interdiction (or downgrading) problem returns a solution that violates the interdiction budget $b$ by a factor of at most $\beta$ and provides a final cut (in the interdicted graph) with cost at most $\alpha$ times the optimal cost of a minimum cut in a solution of interdiction budget at most $b$.
\end{definition}


Chestnut and Zenklusen~\cite{chestnut2017hardness} study the network flow interdiction problem (NFI), which is the undirected and edge interdiction version of WNVIP.
NFI is also known to be essentially equivalent to the Budgeted minimum $st$ cut problem~\cite{papadimitriou2000approximability}.
NFI is also a recasting of the $k$-route $st$-cut problem~\cite{chuzhoy2016,guru2015}, where a minimum cost set of edges must be deleted to reduce the node or edge connectivity between $s$ and $t$ to be $k$.
The results of Chestnut and Zenklusen, and Chuzhoy et al.~\cite{chuzhoy2016} show that an $(\alpha, 1)$-approximation for WNVIP implies a $2(\alpha)^2$-approximation for the notorious Densest k-Subgraph (DkS) problem.
The results of Chuzhoy et al.~\cite{chuzhoy2016} (Theorem 1.9 and Appendix section B) also imply such a hardness for NVIP even with unit edge costs.
Furthermore, Chuzhoy et al.~\cite{chuzhoy2016} also show that there is no $(C,1 + \gamma_C)$-bi-criteria approximation for WNVIP assuming Feige’s Random $\kappa$-AND Hypothesis (for every $C$ and sufficiently small constant $\gamma_C$ ). For example, under this hypothesis, they show hardness of $(\frac{11}{10} - \epsilon,\frac{25}{24} - \epsilon)$ approximation for WNVIP.

Chestnut and Zenklusen give a $2(n-1)$ approximation algorithm for NFI for any graph with $n$ vertices.
In the special case where the graph is planar, Philips~\cite{Phillips:1993} gave an FPTAS and Zenklusen~\cite{Zen10DAM} extended it  to handle the vertex interdiction case.

Burch \textit{et al.}~\cite{burch03} give a $(1+\varepsilon,1), (1, 1+\frac{1}{\varepsilon})$ pseudo-approximation algorithm for NFI. Given any $\varepsilon >0$, this algorithm returns either a $(1+\epsilon)$-approximation, or a solution violating the budget by a factor of
$1+\frac{1}{\epsilon}$ but has a cut no more expensive than the optimal cost. However, we do not know which case occurs \textit{a priori}.
In this line of work, Chestnut and Zenklusen~\cite{chestnut2016interdicting} have extended the technique of Burch et al. to derive pseudo-approximation algorithms for a larger class of NFI problems that have good LP descriptions (such as duals that are box-TDI).
Chuzhoy et al.~\cite{chuzhoy2016} provide an alternate proof of this result by subdividing edges with nodes of appropriate costs.


\subsection*{Our Contributions.}

\begin{enumerate}
\item We define and initiate the study of multi-level node downgrading problems by defining the Network Vertex Leveling Downgrading Problem (NVLDP) and provide the first results for it. This problem extends the study in~\cite{Zen10DAM} of the vertex interdiction problem in order to consider a richer set of interdiction functions.

\item Chuzhoy et al.~\cite{chuzhoy2016}, and Chestnut and Zenklusen \cite{chestnut2017hardness} showed that a $(\alpha(n), 1)$-approximation for NFI leads to a $2(\alpha(n^2))^2$-approximation for the Densest k-Subgraph problem, implying a similar conclusion for WNVIP.
    The results of Chuzhoy et al.~\cite{chuzhoy2016} also imply such a hardness for NVIP even with unit costs.
    Using similar techniques but more directly, we show a $(1, \beta(n))$-approximation for WNVIP also leads to a $2(\beta(n^2))^2$-approximation for DkS.
    These results show that designing unicriterion approximations for WNVIP are at least ``DkS hard".
    Along with the harness of $(\frac{11}{10} - \epsilon,\frac{25}{24} - \epsilon)$ approximation for WNVIP due to Chuzhoy et al.~\cite{chuzhoy2016} mentioned above, these suggest that we focus on bicriteria approximation results.

    We sharpen our $(1, \beta)$ hardness result further and show that it is also ``DkS hard" to obtain a $(1, \beta)$-approximation for NVIP and NVIP with unit costs. Note that this is in sharp contrast to the edge interdiction case. NFI with unitary interdiction cost and unitary cut cost can be solved by first finding a minimum cut and then interdicting $b$ edges in that cut~\cite{Zen14}. (Details are in Appendix~\ref{sec:hardness})

\item  Burch \textit{et al.}~\cite{burch03} gave a polynomial time algorithm that finds a $(1+1/\epsilon, 1)$ or $(1, 1+\epsilon)$-approximation for any $\epsilon >0$ for WNVIP in digraphs. This was reproved more directly by Chuzhoy et al~\cite{chuzhoy2016} by converting both interdiction and arc costs into costs on nodes. We show that this strategy can also be extended to give a simple $(4,4(1+\epsilon))$-bicriteria approximation for the multiway cut generalization in directed graphs and a $(2(1+\epsilon)\ln k, 2(1+\epsilon)\ln k)$-approximation for the multicut vertex interdiction problem in undirected graphs, for any $\epsilon >0$. (Details are in Appendix~\ref{sec:vip})

\item	For the downgrading variant NVDP, we show that the problem of detecting whether there exists a downgrading set that gives a zero cost cut can be solved in polynomial time. We then design a new LP rounding approximation algorithm that provides a $(4, 4)$-approximation to NVDP. We use a carefully constructed auxiliary graph so that the level-cut algorithm based on ball growing for showing integrality of $st$-cuts in digraphs (See. e.g.~\cite{Chuzoy-icam}) can be adapted to choose nodes to downgrade and arcs to cut based on the LP solution. (Section~\ref{sec:nvdp}, Appendix~\ref{app:nvdp})

\item For the most general version NVLDP with $k$ levels of downgrading each vertex and $k^2$ possible costs of cutting an edge, we generalize the LP rounding method for NVDP to give a $(4k,4k)$-approximation (Section~\ref{sec:nvldp}).

\end{enumerate}

%% file: nvdp.tex
\section{Network Vertex Downgrading Problem (NVDP)}
\label{sec:nvdp}


As an introduction and motivation to the LP model and techniques used to solve NVLDP, in this section, we focus on the special case NVDP, where there is only one other level to downgrade each vertex to. Our main goal is to show the following theorem.
\begin{theorem}
	There exists a polynomial time algorithm that provides a $(4, 4)$-approximation to NVDP on an $n$-node digraph.
	\label{NVDPapprox}
\end{theorem}

\paragraph{LP Model for NVDP.}
To formulate the NVDP as a LP, we begin with the following standard formulation for minimum $st$-cuts~\cite{guenin2014gentle}.
\begin{align}
	\min \qquad &\sum_{e\in A(G)}c(e)x_e\notag \\
	\text{s.t.} \qquad &d_v\le d_u +x_{uv} & \forall uv \in A(G) \label{LPtri}\\
	&d_s=0, d_t\ge 1&\notag\\
    & x_{uv} \geq 0 & \forall uv \in A(G)
\end{align}
An integer solution for $st$-cut can be interpreted as setting $d$ to be 0 for nodes in the $s$ shore and 1 for nodes in the $t$ shore of the cut. Constraint (\ref{LPtri}) then insist that the $x$-value for arcs crossing the cut to be set to 1.
The potential $d_v$ at node $v$ can also be interpreted as a distance label starting from $s$ and using the nonnegative values $x_{uv}$ as distances on the arcs.
Any optimal solution to the above LP can be rounded to an optimal integer solution of no greater value by using the $x$-values on the arcs as lengths, growing a ball around $s$, and cutting it at a random threshold between 0 and the distance to $t$ (which is 1 in this case). The expected cost of the random cut can be shown to be the LP value (See e.g.,~\cite{Chuzoy-icam}), and the minimum such ball can be found efficiently using Dijkstra's algorithm. Our goal in this section is to generalize this formulation and ball-growing method to NVDP.

One difficulty in NVDP comes from the fact that every arc has four associated costs and we need to write an objective function that correctly captures the final cost of a chosen cut. One way to overcome this issue is to have a distinct arc associated with each cost. In other words, for every original arc $uv\in A(G)$, we create four new arcs $[uv]_0, [uv]_1, [uv]_2, [uv]_3$ with cost $c_e, c_{eu}, c_{euv}, c_{ev}$ respectively. Then, every arc has its unique cost and it is now easier to characterize the final cost of a cut. We consider the following auxiliary graph. See Figure \ref{auxH}.

\begin{figure}[h]
	\centering
	\includegraphics[scale=0.9]{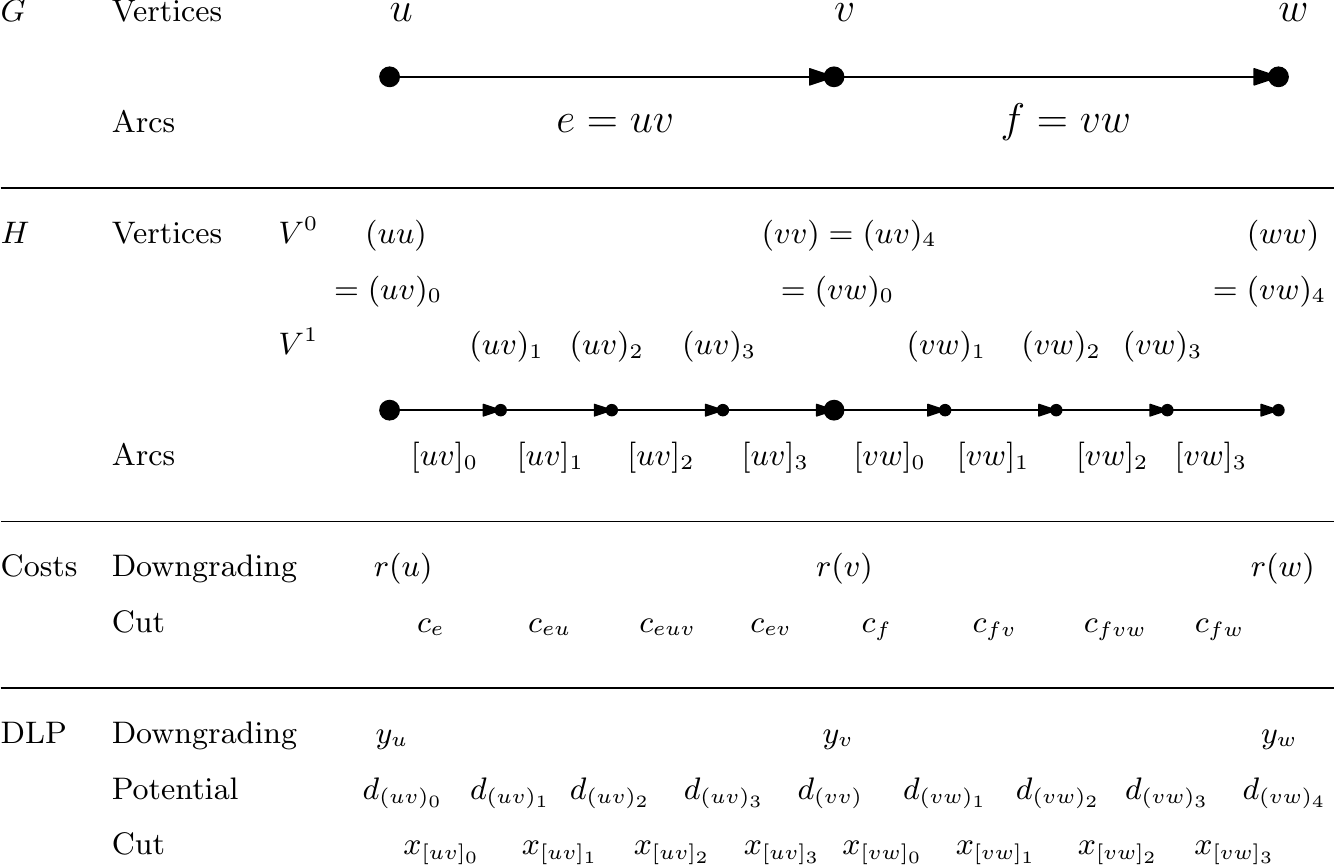}
	\caption{Construction of Auxiliary graph $H$ \label{auxH}}
\end{figure}

\paragraph{Constructing Auxiliary Graph $H$}
Let $V(H)=V^0(H)\cup V^1(H)$ where $V^0(H)=\{(vv): v\in V(G)\}$ and $V^1(H)=\{(uv)_i: uv\in A(G), i=1, 2, 3\}$. Define $A(H)=\{[uv]_0=(uu)(uv)_1, [uv]_1=(uv)_1(uv)_2, [uv]_2=(uv)_2(uv)_3, [uv]_3=(uv)_3(vv): uv\in A(G)\}$. Essentially, the vertices $(uu)\in V^0(H)$ corresponds to the original vertices $u\in V(G)$ and for every arc $uv\in A(G)$, we replace it with a path $(uu)(uv)_1(uv)_2(uv)_3(vv)$ where the four arcs on the path are $[uv]_0, [uv]_1, [uv]_2, [uv]_3$. For convenience and consistency in notation, we define $(uv)_0 := (uu), (uv)_4:=(vv)$. Note that the vertices of $H$ will always be denoted as two lowercase letters in parenthesis while arcs in $H$ will be two lowercase letters in square brackets with subscript $i=0, 1, 2, 3$. The cost function $c:A(H)\to \mathbb{R}_{\ge 0}$ is as follows: $c([uv]_0)=c_e, c([uv]_1)=c_{eu}, c([uv]_2)=c_{euv}, c([uv]_3)=c_{ev}$. Since we can only downgrade vertices in $V^0$, to simplify the notation, we retain $r(v)$ as the cost to downgrade vertex $(vv)\in V^0$. Note that $|V(H)|=3|A(G)|+|V(G)| = O(n+m)$.

\paragraph{Downgrading LP on $H$.}
	Given the auxiliary graph $H$, we can now construct an LP similar to the one for $st$-cuts. For vertices $(vv)\in V^0(H)$ corresponding to original vertices of $G$, we define a downgrading variable $y_v$ representing whether vertex $v$ is downgraded or not in $G$. For every arc $[uv]_i\in A(H)$, we have a cut variable $x_{[uv]_i}$ to indicate if the arc belongs in the final cut of the graph. Lastly for all vertices $(uv)_i\in V(H)$, we have a potential variable $d_{(uv)_i}$ representing its distance from the source $(ss)$.
	
	The idea is to construct an LP that forces $s, t$ to be at least distance $1$ apart from each other as before. This distance can only be contributed from the arc variables $x_{[uv]_i}$. The downgrading variables $y_v$ imposes limits on how large these distances $x_{[uv]_i}$ of some of its incident arcs can be. The motivation is that the larger $y_u$ and $y_v$ are, the more we should allow arc $[uv]_2$ to appear in the final cut over the other arcs $[uv]_0, [uv]_1, [uv]_3$ in order to incur the cheaper cost of $c_{euv}$. Now consider the following downgrading LP henceforth called DLP.

	\begin{align}
	\qquad \min \qquad &\sum_{[uv]_i\in A(H)}c([uv]_i)x_{[uv]_i}\notag \\
	\text{s.t.} \qquad &\sum_{(vv)\in V^0(H)}r(v)y_v\le b &  \label{DLPb}\\
	&d_{(uv)_{i+1}}\le d_{(uv)_i} +x_{[uv]_i} & \forall \text{arc } [uv]_i, 0\le i\le3  \label{DLPtri}\\
	&x_{[uv]_2} + x_{[uv]_3} +x_{[vw]_1}+ x_{[vw]_2} \le y_v & \forall \text{ path } (uv)_3(vv)(vw)_1 \label{DLPint}\\
	&d_{(ss)}=0, d_{(tt)}= 1 , y_s=0,  y_t=0	\notag
	\end{align}

Figure \ref{auxH} also includes the list of variables associated with $H$.  Our objective is to minimize the cost of the final cut. Constraint (\ref{DLPb}) corresponds to the budget constraint for the downgrading variables. Constraint (\ref{DLPtri}) corresponds to the triangle inequality for each arc in $H$, analogous to the triangle inequality Constraint (\ref{LPtri}) in the LP for min-cuts.
	
	Constraint (\ref{DLPint}) relates cut and downgrade variables. If we do not consider any constraints related to downgrading variables for a moment, the LP will naturally always want to choose the cheapest arc $[uv]_2$ over $[uv]_0, [uv]_1, [uv]_3$ when cutting somewhere between $(uu)$ and $(vv)$. However, the cut should not be allowed to go through $[uv]_2$ if one of $u, v$ is not downgraded. In other words $x_{[uv]_2}$ should be at most $y_u, y_v$. This reasoning gives the constraint $x_{[uv]_2}, x_{[uv]_3}, x_{[vw]_1}$ and, $x_{[vw]_2}$ all needs to be $\le y_v$ for in-arcs $uv$ and out-arcs $vw$. Now consider an arc $f=vw\in E(G)$. In an integral solution, if $v$ is downgraded, the arc $vw$ incurs a cost of either $c_{fv}$ or $c_{fvw}$ but not both, since $v$ must lie on one side of the cut. This translates to a LP solution where only one of the arcs $[vw]_2, [vw]_3$ is in the final cut. Then, the $x$ variables corresponding to $[vw]_1+ [vw]_2$ is at most $1$. Thus, a better constraint to impose is $x_{[vw]_1}+x_{[vw]_2}\le y_v$. To push it further, consider a path $uvw$ in $G$. In an integral solution, at most one of the arcs $uv, vw$ appears in the final cut. This implies if $v$ is downgraded, only one of the costs $c_{ev}, c_{euv}, c_{fv}, c_{fvw}$ is incurred. This corresponds to the tighter constraint (\ref{DLPint}). Note that for every vertex $v\in V(G)$, for every pair of incoming and outgoing arcs of $v$, we need to add one such constraint. Then for every vertex in $G$, we potentially have to add up to $n^2$ many constraints. In total, the number of constraints would still only be $O(n^3)$. The last few constraints in DLP make sure $s$ and $t$ are $1$ distance apart and cannot themselves be downgraded.

	
The validity of this LP model is shown by proving that a solution of NVDP corresponds to a feasible solution of DLP and an integral solution of DLP can be translated to a feasible solution of NVDP with the same or even better cost. These proofs are included in Appendix \ref{app:nvdp}.
	
\paragraph{Bicriteria Approximation for NVDP.}
We now prove Theorem~\ref{NVDPapprox}.
We will work with an optimal solution from DLP on the auxiliary graph $H$.
The idea is to use a ball-growing algorithm that greedily find cuts until one with the promised guarantee is produced. The reason this algorithm is successful is proved by analyzing a randomized algorithm that picks a number $0\le \alpha\le 1$ uniformly at random and chooses a cut at distance $\alpha$ from the source $s$. Then we choose vertices to downgrade and arcs to cut based on arcs in this cut at distance $\alpha$. By computing the expected downgrading cost and the expected cost of the cut arcs, the analysis will show the existence of a level cut that satisfies our approximation guarantee.

To achieve the desired result, we cannot work with the graph $H$ directly. This is because the ball-growing algorithm works if the probability of cutting some arc can be bounded within some range. This bound exists for the final cut arcs but not for the final downgraded vertices. Consider a vertex $v$; it is downgraded if any arcs of the form $[uv]_2, [uv]_3, [vw]_1, [vw]_2$ is cut in $H$. Thus it has potential of being cut anywhere between $(uv)_2$ and $(vw)_3$. We would like to use Constraint (\ref{DLPint}) to bound this range but we cannot since we do not know how long the arc $[vw]_0$ is. Thus we will contract some arcs first in order to properly use Constraint (\ref{DLPint}).

Let $(x^*, y^*, d^*)$ be an optimal solution to DLP where the optimal cost is $c^*$. It follows from the validity of our model that $c^*$ is at most the cost of an optimal integral solution.

\paragraph{Constructing Graph $H'$}
	For every arc $uv\in A(G)$, we compare the value $x^*_{[uv]_0}$ and $x^*_{[uv]_1}+x^*_{[uv]_2}+x^*_{[uv]_3}$. The reason we separate this way is because the variables in the second term are influenced by the downgrading values on $u, v$. Thus the more we downgrade $u$ and $v$, the larger we are allowed to increase the second sum, the more distance we can place between $u$ and $v$. For an arc $uv\in A(G)$, if $x^*_{[uv]_0} < x^*_{[uv]_1}+x^*_{[uv]_2}+x^*_{[uv]_3}$, we say $uv$ is an \textit{aided} arc since most of its distance is contributed by the downgrading values on the $u, v$ and thus the downgrading values help to reduce its cost. For all other arcs, we say $uv$ is an \textit{unaided} arc since most of the distance would be contributed by the arc $[uv]_0$, corresponding to simply paying for the original cost of deletion $c_e$ without the help from downgrading. To construct $H'$, if $uv$ is an aided arc, then contract $[uv]_0$. Otherwise, contract $[uv]_1, [uv]_2, [uv]_3$.

	Consider a path $P=(uv)_0(uv)_1(uv)_2(uv)_3(uv)_4$ in $H$. Note that the length of this path is shortened in $H'$ depending on whether $uv$ is an aided/unaided arc. However, since we always retain the larger of $x^*_{[uv]_0}$ and $x^*_{[uv]_1}+x^*_{[uv]_2}+x^*_{[uv]_3}$ in $H'$,  the path's length is reduced by at most one half. Then it follows that the distance between any two vertices in $H'$ is reduced by at most one half. In particular, it follows that the shortest path between the source and the sink is at least $1/2$.
This property will be crucial in arguing that the solution chosen by our algorithm has low cost relative to the LP optimum.


\begin{algorithm}[h]	
	\caption{Ball-Growing Algorithm for NVDP}\label{alg:sp}
	\begin{algorithmic}[1]
		\REQUIRE a graph $G$ and its axillary graph $H'$ with non-negative arc-weights $x^*_{[uv]_i}$, source $(ss)$, sink $(tt)$, arc cut costs $c([uv]_i)$ and vertex downgrading costs $r(v)$,
		\ENSURE a vertex set $V'$ and an arc cut $E'$ of $G$ such that $\Sigma_{v\in V'} r(v)\le 4b, c^{V'}(E')\le 4c^*$
		\STATE initialization $V=\{(ss)\}, D((uv)_i)=0$ for all $(uv)_i\in V(H')$,
		\REPEAT
		\STATE let $X'\subseteq A(H')$ be the cut induced by $V$,
		\STATE find $[uv]_i=(uv)_i(uv)_{i+1} \in X'$ minimizing $D((uv)_i)+x^*_{[uv]_i}$,
		\STATE update by adding $(uv)_{i+1}$ to $V$, update $D{(uv)_{i+1}}=D((uv)_i)+x^*_{[uv]_i}$,
		\STATE let $E'=\{uv\in A(G): \{[uv]_0, [uv]_1, [uv]_2, [uv]_3\}\cap X'\neq \emptyset\}$ and $V'=\{v\in V(G): \{[uv]_2, [uv]_3, [vw]_1, [vw]_2\}\cap X'_k\neq\emptyset \text{ for some } u, w\in V(G)\}$
		\UNTIL $\Sigma_{v\in V'} r(v)\le 4b$ and $c^{V'}(E')\le 4c^*$
		\STATE output the set $V', E'$.
	\end{algorithmic}
\end{algorithm}

In Algorithm~\ref{alg:sp}, it follows the general ball-growing technique and looks at cuts $X'$ at various distances from the source. Note that the algorithm adds at least one vertex to a node set $V$ at each iteration so it runs for at most $|V(H')|=O(m)$ steps when applied to the graph $H'$, where $m$ denotes the number of arcs in the original graph $G$.

At each iteration, it associates $E'$, the set of original arcs corresponding to those in $X'$ and a vertex subset $V'$, representing the set of vertices we should downgrade based on the arcs in $X$. For example, if $[uv]_2\in X'_k$, then we should downgrade both $u, v$. Note that since $X'$ is a cut in $H'$, it follows that $E'$ is also a cut in $G$.

First, we show that there exists a cut $X'$ at some distance $\alpha$ from the source such that the associated sets $V', E'$ provides the approximation guarantee.

\begin{lemma}
\label{lem:nvdp}
	There exists $X', V', E'$ such that $\Sigma_{v\in V'} r(v)\le 4b, c^{V'}(E')\le 4c^*$
\end{lemma}

The main idea of the proof is to pick a distance uniformly at random and study the cut at that distance. Then, argue that the extent to which an arc is cut (chosen in $E'$ above) in the random cut is at most twice its $x^*$-value, using the property of the model.
For this, we crucially use the fact that the random cut distance is chosen between zero and at least half since even after the arc contractions, the distance of $t$ from $s$ is still at least half. When nodes are chosen in the random cut (in $V'$ above) to be downgraded, we use constraint (\ref{DLPint}) in the LP along with the properties of the algorithm to argue similarly that the probability of downgrading a node in the process is at most twice its $y^*$-value.
To obtain a cut where we simultaneously do not exceed both bounds, we use Markov's inequality to show a probability of at least half of being within twice these respective expectations, hence giving us a single cut with both bounds within four times their respective LP values. Details of the proof of the lemma  are in Appendix~\ref{app:ipco}.

Lastly, in order to prove the validity of Algorithm \ref{alg:sp}, we need to show it can eventually find the correct cut at distance $\alpha$.

\begin{lemma}
	Let $\alpha$ be the smallest value such that the associated cut $X', V', E'$ at distance $\alpha$ provides the promised apporximation guarantee. Then Algorithm \ref{alg:sp} can find it by checking all cuts at distance $\alpha'\le\alpha$.
	\label{alg-lemma}	
\end{lemma}

	The proof is included in Appendix \ref{app:ipco}. Then, Theorem \ref{NVDPapprox} is proved by simply running Algorithm \ref{alg:sp} on the auxiliary graph $H'$.

%% file: app-ipco.tex
\section{Validity of Algorithm ~\ref{alg:sp}} 
\label{app:ipco}

\begin{proof}[Proof of Lemma \ref{lem:nvdp}]
Let$D((uv)_i)$ be the distance from the source $(ss)$ to any vertex $(uv)_i\in V(H')$ viewing the $x^*$ variables as distances. Note that $D((tt))\ge 1/2$ since the original distance is at least $1$ and $H'$ reduces the distance by at most $1/2$. Note that the triangle-inequality holds under this distance metric where $D((uv)_i)-D(('v')_{i'})$ is at most the distance between $(uv)_i$ and $(u'v')_{i'}$.

\paragraph{Defining the Random Variables}
Let $\alpha$ be chosen uniformly at random from the interval $[0, D((tt))]$. Consider $X_\alpha:= \{[uv]_i\in A(H'): D((uv)_i)\le \alpha <D((uv)_{i+1})\}$, the cut at distance $\alpha$ in $H'$. Let $E_\alpha=\{uv\in A(G): [uv]_i\in X_\alpha \text{ for some } i=0, 1, 2, 3\}$, representing the original arcs corresponding to those in $X_\alpha$. Let $V_\alpha =\{v\in V(G): \{[uv]_2, [uv]_3, [vw]_1, [vw]_2\}\cap X_\alpha \neq \emptyset \text{ for some } u, w\in V(G)\}$, representing the set of vertices we should downgrade so that the final cost of the arcs $E_\alpha$ matches the cost associated to $X_\alpha$. More precisely, we want $c^{V_\alpha}(E_\alpha)=\Sigma_{[uv]_i\in X_\alpha} c([uv]_i)$. Note that by construction $E_\alpha$ is a $st$-cut in $G$. Let $\mathcal{V}=\Sigma_{v\in V_\alpha} r(v), \mathcal{E}=c^{V_\alpha}(E_\alpha)$. Our goal is to show that these two random variables $\mathcal{V, E}$ have low expectations and obtain our approximation guarantee using Markov' inequality.
In particular, we will prove that $\mathbb{E}[\mathcal{V}] \leq 2b$, and that $ \mathbb{E}[\mathcal{E}] \leq 2c^*$ where $c^*$ is the optimal value of DLP.

To understand $\mathcal{E}$, for every arc $e=uv\in A(G)$, we introduce the indicator variables $\mathcal{E}_e$ to be $1$ if arc $e\in E_\alpha$ and $0$ otherwise. Then $\mathcal{E}=\Sigma_{e\in A(G)} \mathcal{E}_e c^{V_\alpha}(e)$. To study the value of $\mathcal{E}_ec^{V_\alpha}(e)$, we can break into several cases depending on which arc $[uv]_i\in X_\alpha$. Note that if $[uv]_i\notin X_\alpha$ for $i=0, 1, 2, 3$, then $e\notin E_\alpha$ and $\mathbb{E}_ec^{Y^*}(e)=0$.
Next, if we assume $[uv]_i\in X_\alpha$, then one can check that $c^{V_\alpha}(e)\le c([uv]_i)$. The proof of this claim is similar to the proof of Claim \ref{comparec} where we check the different cases depending on whether $u, v\in V_\alpha$.

Slightly abusing the notation, define the indicator variable $\mathcal{E}_{[uv]_i}$ for arc $[uv]_i\in A(H)$ to be $1$ if $[uv]_i\in X_\alpha$ and 0 otherwise. Then, we can upper-bound the expectation of $\mathcal{E}$ using conditional expectations of the events $\mathcal{E}_{[uv]_i}=1$ as follows.
\begin{align*}
\mathbb{E}[\mathcal{E}] =  &\Sigma_{e\in A(G)} \mathbb{E}[\mathcal{E}_e c^{V_\alpha} (e)]\\
 =& \Sigma_{e\in A(G)} \Sigma_{i=0}^3 \mathbb{E} [c^{v_\alpha}(e)| \mathcal{E}_{[uv]_i} =1 ] \cdot Pr[\mathcal{E}_{[uv]_i}=1] \\
 \le & \Sigma_{e\in A(G)}\Sigma_{i=0}^3 c([uv]_i) Pr[\mathcal{E}_{[uv]_i}=1]
\end{align*}
To understand the probability of $\mathcal{E}_{[uv]_i}=1$, note that an arc $[uv]_i\in X_\alpha$ if and only if $D((uv)_i)\le \alpha < D((uv)_{i+1})$. Then, $Pr[[uv]_i\in \chi(\alpha)]\le (D((uv)_{i+1})-D((uv)_i)/D((tt))\le 2x^*_{[uv]_i}$ since $D((tt))\ge 1/2$. Combining with the previous inequalities, we see that
\begin{align*}
 \mathbb{E}[\mathcal{E}] \le &\Sigma_{uv\in A(G)}\Sigma_{i=0}^3 c([uv]_i) Pr[\mathcal{E}_{[uv]_i}=1]\\
 \le & \Sigma_{uv\in A(G)}\Sigma_{i=0}^3 c([uv]_i)2 x^*_{[uv]_i}= 2 c^*.
\end{align*}

Next, we show similar result for $\mathcal{V}$. Note that $\mathbb{E}[\mathcal{V}]=\Sigma_{v\in V(G)}r(v)\cdot Pr[v\in V_\alpha]$. Recall that $v\in V_\alpha$ if and only if there exists a vertex $u$ or $w$ such that at least one of $[uv]_2, [uv]_3, [vw]_1, [vw]_2\in X_\alpha$. Note that if $uv$ is an unaided arc, then $[uv]_2, [uv]_3$ would have been contracted in $H'$ and would never be chosen in $X_\alpha$. Therefore, we only need to consider aided arcs. In order to upper-bound the probability of choosing $v$ into $V_\alpha$, we thus need to find the range of possible $\alpha$ that might affect $v$. For any vertex $v\in V(G)$, it follows that we only need to examine aided arcs incident to the vertex $v$. Let $u\in V(G)$ such that $uv\in A(G), uv$ is an aided arc and $D((uv)_2)$ is minimum. Let $w\in V(G)$ such that $vw$ is an aided arc and $D((vw)_3)$ is maximum. By definition, all arcs of the form $[zv]_2, [zv]_3, [vz']_1, [vz']_2$ lie between the range $D((uv)_2), D((vw)_3)$. Then, $v$ is chosen only if $\alpha$ is between $D((uv)_2)$ and $D((vw)_3)$. Note that the distance between $(uv)_2$ and $(vw)_3$ is upper-bounded by the length of a shortest path in $H'$. Note that since $vw$ is an aided arc, $[vw]_0$ is contracted in $H'$. Then $(uv)_2(uv)_3(vv)(vw)_1(vw)_2$ is a path in $H'$. Thus $D((vw)_3)-D((uv)_2)\le x^*_{[uv]_2}+x^*{[uv]_3}+x^*{[vw]_1}+x^*_{[vw]_2}\le y^*_v$ where the last inequality follows from Constraint (\ref{DLPint})\footnote{This is the main reason why we distinguish between aided and unaided arcs and contract the appropriate one to construct $H'$. Without the contraction, the distance between $(uv)_2$ and $(vw)_3$ includes the arc $[vw]_0$ and thus could be arbitrarily larger than $y^*_v$.}. Then, $Pr[D((uv)_2)\le \alpha < D((vw)_3)] \le y^*_v/D((tt)) \le 2y^*_v$. Therefore
\begin{align*}
	\mathbb{E}[\mathcal{V}]= & \Sigma_{v\in V(G)} r(v)\cdot Pr[v\in V_\alpha] \\
	\le & \Sigma_{v\in V(G)} r(v) 2y^*_v\le 2b.
\end{align*}

	Lastly, by Markov's inequality, $Pr[\mathcal{V}\le 4b]\ge 1/2, Pr[\mathcal{E}\le 4c^*]\ge 1/2$. Then it follows there exists $0\le \alpha\le D((tt))$ such that $\Sigma_{v\in V_\alpha} r(v)\le 4b$ and $c^{V_\alpha}(E_\alpha)\le 4c^*$, proving our lemma.

\end{proof}

\begin{proof} [Proof of Lemma~\ref{alg-lemma}]
	Let $\beta<\beta'\le \alpha$ and consider the cuts $X_\beta, X_{\beta'}$ at distance $\beta, \beta'$ respectively. Let $Y_\beta, Y_{\beta'}$ be the vertices of the connected component  of $H'\backslash X_\beta, H'\backslash X_{\beta'}$ containing $(ss)$ respectively. By definition, all vertices in $Y_\beta, Y_{\beta'}$ are at distance at most $\beta, \beta'$ from $(ss)$ respectively. Thus, $Y_\beta\subseteq Y_{\beta'}$. Then $X_\beta, X_{\beta'}$ are distinct cuts if and only if $Y_\beta$ is a proper subset of $Y_{\beta'}$. This implies that the cuts at distance at most $\alpha$ can be properly ordered and there are only polynomially many such cuts. Then, we will use induction to prove our claim that Algorithm \ref{alg:sp} checks all cuts in this order. 
	
	The base case where $\beta=0$ is obviously true.  Now suppose $X_\beta, X_{\beta'}$ are two consecutive cuts in the ordering and the algorithm just checked $X_\beta$. Let $Y'=Y_{\beta'}\backslash Y_\beta$ and let $F'$ be the set of all arcs between $Y_\beta$ and $Y'$. First, note that all arcs in $F'$ has the same length, otherwise there exists another cut at distance strictly between $\beta$ and $\beta'$ (by adding the other end of the shortest arc in $F'$ into the set $Y$ and looking at the cut it induces). Then, the algorithm first pick up any arc in $F'$ and extends the set $Y$ a bit. Then, it continues to to pick up all other arcs in $F'$ before choosing any other arcs due to Step 4 of Algorithm \ref{alg:sp}. Thus, it will eventually also reach and check $X_{\beta'}$. 
	
\end{proof}

%% file: app-nvdp.tex
\section{Appendix for NVDP}
\label{app:nvdp}

\subsection{Detecting Zero in NVDP in Polynomial Time}
In this subsection, we provide an algorithm to detect, in a given instance of NVDP, whether there exists nodes to downgrade such that the downgrading cost is less than the budget and the min cut after downgrading is zero.

In order to demonstrate the main idea of the proof, we first work on a special case of NVDP. Suppose for every arc $e=uv$, $c_e=c_{eu}=c_{ev}=1$ and $c_{euv}=0$. In other words, every arc is unit cost and requires the downgrading of both ends in order to reduce the cost down to zero. For every vertex $v\in V(G)$, we assume the interdiction cost $r(v)=1$. We call this the \textbf{Double-Downgrading Network Problem} (DDNP). We first prove the following.

\begin{lemma}
	Given an instance of DDNP on graph $G$ with budget $b$, there exist a polynomial time algorithm to determine if there exists $Y\subseteq V(G)$ and a $st$-cut $F\subseteq A(G)$ such that $|Y|\le b$ and $c^Y(F)=0$.
\end{lemma}

\begin{proof}
	Let $X\subseteq V(G)$ be a minimum set of vertices to downgrade such that the resulting graph contains a cut of zero cost. Let $F$ be the set of arcs in the graph induced by $X$ (i.e., with both ends in $X$). Note that $F$ are the only arcs with cost zero and hence $F$ is an arc cut in $G$. Furthermore, since $X$ is optimal, $X$ is the set of vertices incident to $F$ (there are no isolated vertices in the graph induced by $X$). Let $V_s, V_t$ be the set of vertices in the component of $G\backslash F$ that contains $s, t$ respectively.
	
	Consider the graph $G^2$ where we add arc $uv$ to $G$ if there exists $w\in V(G)$ such that $uv, vw\in A(G)$. First we claim that $X$ is a vertex cut in $G^2$. Suppose there is an $st$ path in $G^2 \setminus X$ where the first arc crossing over from $V_s$ to $V_t$ is $uv$. Note that any such $u \in V_s \setminus X$ and $v \in V_t \setminus X$ are distance 3 apart and hence do not have an arc between them in $G^2$, a contradiction.

	Given any vertex cut $Y$ in $G^2$, we claim that downgrading $Y$ in $G$ creates a $st$-cut of zero cost, by deleting the arcs induced by $Y$ from $G$. Suppose for a contradiction there is an $st$-path after downgrading $Y$ and deleting the zero-cost arcs induced by $Y$. Then the path cannot have two consecutive nodes in $Y$. Let $y \in Y$ be a single node along the path with neighbors $y^-,y^+ \not\in Y$. Note that $(y^-,y^+) \in G^2$, and shortcutting over all such single node occurrences from $Y$ in the path gives us a $st$-path in $G^2 \setminus Y$, a contradiction.
	
	This proves that a minimum size downgrading vertex set $Y$ in $G$ whose downgrading produces a zero-cost $st$-cut is also a minimum vertex-cut in $G^2$. Then, one can check if a zero-cut solution exists with budget $b$ for DDNP by simply checking if the minimum vertex-cut in $G^2$ is at most $b$.
\end{proof}

Now, to prove Theorem \ref{NVDPdetect}, we have to slightly modify the graph $G$ and the construction of $G^2$ in order to adapt to the various costs. Our goal is still to look for a minimum vertex cut in an auxiliary graph using $r(v)$ as vertex cost.

\begin{proof}
	Given an instance of NVDP on $G$ with a budget $b$, vertex downgrading costs $r(v)$ and arc costs $c_e, c_{eu}, c_{ev}, c_{eu}$, consider the following auxiliary graph $H$. First, we delete any arc $e$ where $c_e=0$ since they are free to cut anyways. For every arc $e=uv$ where $c_{euv}>0$, subdivide $e$ with a vertex $t_e$ and let $r(t_e)=\infty$. In some sense, since $c_e, c_{eu}, c_{ev}\ge c_{euv}>0$, downgrading $u, v$ cannot reduce the cost of $e$ to zero. Then, we should never be allowed to touch the vertex $t_e$. Let $T$ be the set of all newly-added subdivided vertices. To finish constructing $H$, our next step is to properly simulate $H^2$.
	
	We classify arcs into five types based on which of its costs are zero. Note that we no longer have any arcs where $c_e=0$. Let $A_0:=\{ e=uv: c_{eu}=c_{ev}=c_{euv}=0\}$, the arcs where downgrading either ends reduce its cost to zero. Let $A_{l0}:= \{e=uv: c_{eu}=c_{euv}=0, c_{ev}>0\}, A_{0r}:= \{e=uv: c_{ev}=c_{euv}=0, c_{eu}>0\}, A_{l0r}:= \{e=uv: c_{euv}=0, c_{eu}, c_{ev}>0\}$ respectively represent arcs that require the downgrading of its left tail, its right head, or both in order to reduce its cost. Let $A_1$ be all remaining arcs, those incident to the newly subdivided vertex $t_e$. Now, for every path $uvw$ of length two, we consider adding the arc $uw$ based on the following rules (see Figure \ref{nvdp0detex} for example of newly added arcs):

	\begin{center}
		\begin{tabular}{|c||c|c|c|c|c|}
			\hline
			\multicolumn{6}{|c|}{If $v\notin T$}\\ \hline
			Add $uw$? & \multicolumn{5}{c|}{$vw\in$}\\ \hline \hline
			$uv\in$ & $A_0$ & $A_{l0}$ & $A_{0r}$ & $A_{l0r}$ & $A_1$\\ \hline \hline
			$A_0$ & No & No & No &No & No\\ \hline
			$A_{l0}$ & No & No & Yes & Yes & Yes\\ \hline
			$A_{0r}$ & No & No & No & No & No \\ \hline
			$A_{l0r}$ & No & No & Yes & Yes & Yes \\ \hline
			$A_1$ & No & No & Yes & Yes & Yes\\ \hline \hline
			\multicolumn{6}{|c|}{If $v=t_e\in T$, do not add $uw$}\\
			\hline
		\end{tabular}
	\end{center}
	
	\begin{figure}
		\centering
		\includegraphics{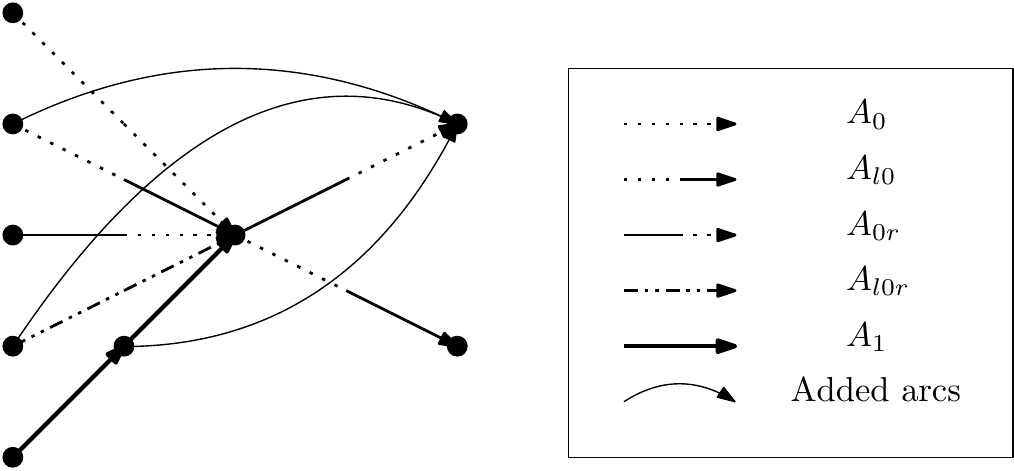}
		\caption{Example of Added Arcs in $H$ \label{nvdp0detex}}
	\end{figure}

	The idea is similar to the proof for DDNP. If $uv, vw\in A_{l0r}$, downgrading $v$ is not enough to cut $uv, vw$ for free. Thus we add arc $uw$ to keep the connectivity. If $uv\in A_{0r}$, then downgrading $v$ should reduce the cost of $uv$ to 0. Thus, we do not want to bypass $v$ by adding an arc $uw$. If $v=t_e\in T$, since $r(v)$ has high cost, we never cut it so we do not need to strengthen the connectivity by adding arcs $uw$.
	
	Let $(X, F)$ be a solution to NVDP where $\Sigma_{v\in X} r(v)$ is minimum, $F$ is a $st$-cut and $c^X(F)=0$. Let $V_s$ be all vertices connected to $s$ in $G\backslash F$. We claim that $X$ is a vertex cut in $H$. Suppose not and there exists a $st$-path in $H$ and let $uv$ be the first arc of the path leaves $V_s$. If $v=t_e\in T$, then arc $e\in F$, contradicting $c^X(F)=0$. If $uv\in A(G)$, then $uv\in F$. Since $u, v\notin X$, $c^X(uv)>0$, a contradiction. If $uv$ is a newly added arc, then there exist $v'\in V(G)$ such that $uv'v$ is a path in $G$. By definition, $V_s\cap T=\emptyset$ so $u, v\notin T$. Then, there are only four cases where we add arc $uw$ to create $H$. In all cases, downgrading $v'$ does not reduce the cost of $uv'', v'v$ to $0$. Since at least one of $uv', v'v\in F$, it contradicts $c^X(F)=0$.
	
	Given a minimum vertex cut $Y$ in $H$, we claim downgrading $Y$ in $G$ creates a $st$-cut of zero cost. Note that $Y\cap T=\emptyset$ since any vertex in $T$ is too expensive to cut. Suppose for a contradiction there is an $st$-path $P$ that does not cross an arc with cost $0$ after downgrading $Y$. Let $P'$ be the corresponding path in $H$. If $P$ contains two consecutive vertices $u, v\in Y$, then $c_{euv}>0$ and it would have been subdivided. This implies there are no consecutive vertices of $Y$ in $P'$. Let $uvw$ be a segment of $P'$ where $v\in Y$. Since downgrading $v$ does not reduce its incident arcs to a cost of $0$, it follows that $uv\in A_{l0}\cup A_{l0r}\cup A_1$ and $vw\in A_{0r}\cup A_{l0r}\cup A_1$. Then it follows that $uw\in A(H)$. Then, every vertex $v\in Y\cap V(P')$ can be bypassed, a contradiction.
	
	This implies a minimum vertex cut in $H$ is a downgrading set that creates a zero-cost cut in $G$. Then, by checking the min-vertex cut cost of $H$, we can determine whether a zero-solution exists for $G$ with budget $b$.
\end{proof}

\subsection{Validity of DLP}

The following lemmas shows the validity of our defined DLP for NVDP.

\begin{lemma}
	An optimal solution to NVDP provides a feasible integral solution to DLP with the same cost.
\end{lemma}

\begin{proof}
	Given a digraph $G$ with cost functions $c_e, c_{eu}, c_{ev}, c_{euv}$, a source $s$ and a sink $t$, let $Y\subseteq V(G), F\subseteq A(G)$ be an optimal solution to NVDP where $r(Y)\le b, F$ is a $st$-cut and $c^Y(F)$ is minimum. Then, a feasible solution $(x, y, d)$ to DLP on the graph $H$ can be constructed as follows:
	
	\begin{itemize}
		\item For the cut variables $x$, let
		\begin{itemize}
			\item $x_{[uv]_0}=1$ if $uv\in F$ and $u, v\notin Y$, 0 otherwise,
			\item $x_{[uv]_1}=1$ if $uv\in F$ and $u\in Y, v\notin Y$, 0 otherwise,
			\item  $x_{[uv]_2}=1$ if $uv\in F, u, v\in Y$ , 0 otherwise,
			\item $x_{[uv]_3}=1$ if $uv\in F, u\notin Y, v\in Y$, 0 otherwise.
		\end{itemize}
		\item For the downgrading variables $y$, let $y_u=1$ if $u\in Y$, 0 otherwise.
		\item For the potential variables $d$, let $d_{[uv]_i}=0$ if $[uv]_i\in S$ and $1$ otherwise,
	\end{itemize}
	where we define $S, T$ as follows.
	Let $F^*$ be the set of arcs in $H$ whose $x$ variable is $1$. We claim that $F^*$ is a $st$-cut in $H$. Note that every $st$-path $Q$ in $H$ corresponds to a $st$-path $P$ in $G$. Then, there is an arc $uv$ in $P$ that is also in $F$. Then, it follows from construction that the $x$ value for one of $[uv]_0, [uv]_1, [uv]_2, [uv]_3$ is 1 and thus there exists $i=0, 1, 2, 3$ such that $[uv]_i\in F^*$. Note that $[uv]_i$ is also in $Q$. Therefore $F^*$ is a $st$-cut in $H$. Then, let $S$ be the set of vertices in $H\backslash F^*$ that is connected to the source $s$ and let $T=V(H)\backslash S$.

	Note that by construction, $(x, y, d)$ is integral and is a feasible solution to DLP. The final objective value $\Sigma_{[uv]_i\in A(H)} c([uv]_i)x_{[uv]_i}=\Sigma_{[uv]_i\in F^*} c([uv]_i)$ and by construction, it matches the cost $c^Y(F^*)$.
\end{proof}

\begin{lemma}
	An integral solution $(x^*, y^*, d^*)$ to DLP with objective value $c^*$ corresponds to a feasible solution $(Y^*, E^*)$ to NVDP such that $c^{Y^*}(E^*)\le c^*$.
	\label{valid2}
\end{lemma}

\begin{proof}
	Given a directed graph $G$ and its auxiliary graph $H$, let $(x^*, y^*, d^*)$ be an optimal integral solution to DLP with an objective value of $c^*$. Let $F^*\subseteq A(H)$ be the set of arcs whose $x^*$ value is 1. Let $Y^*\subseteq V^0(H)$ whose $y$ value is 1. Let $E^*\subseteq A(G) =\{uv\in A(G): [uv]_i\in F^* \text{ for some } i=0, 1, 2, 3\}$ be the set of original arcs of those in $F^*$.
	
	Note that by construction, $Y^*$ does not violate the budget constraint. Every $st$-path in $G$ corresponds directly to a $st$-path in $H$. Since $F^*$ is a $st$-cut in $H$, it follows that $E^*$ is a $st$-cut in $G$. Then it remains to show that $c^*\ge c^{Y^*}(E^*)$.
	
	Note that
	$$c^*=\Sigma_{[uv]_i\in A(H)} c([uv]_i)x^*_{[uv]_i}= \Sigma_{e=uv\in A(G)}c_ex^*_{[uv]_0}+c_{eu}x^*_{[uv]_1}+ c_{euv}x^*_{[uv]_2}+c_{ev}x^*_{[uv]_3}.$$
	Meanwhile, note that $c^{Y^*}(E^*)=\Sigma_{e=uv\in A(G)} c^{Y^*}(e)$. Thus, it suffices to prove the following claim.
	
	\begin{claim}
		For every arc $e=uv\in A(G)$, $\Sigma_{i=0}^3 c([uv]_i) x^*_{[uv]_i} \ge c^{Y^*}(e)$.
		\label{comparec}
	\end{claim}
	
	First, note that if $e=uv\notin E^*$, then $c^{Y^*}(e)=0$ by definition fo $E^*$. Then, the inequality is trivially true. Thus, we may assume $uv\in E^*$ which implies there exist $i=0, 1, 2, 3$ such that $[uv]_i\in F^*$ and $x^*_{[uv]_i}=1$. We will now break into cases depending on whether $u, v\in Y^*$.
	
	Suppose $u, v\notin Y^*$. Then, $y^*_u=y^*_v=0$ and by constraint (\ref{DLPint}) in DLP, it follows that the $^*x$ value for $[uv]_1, [uv]_2, [uv]_3$ are all $0$. Then, $[uv]_0\in F^*$ and $\Sigma_{i=0}^3 c([uv]_i)x^*_{[uv]_i}=c_e=c^{Y^*}(e)$.
	
	Now, assume $u\in Y^*, v\notin Y^*$. By constraint (\ref{DLPint}), $x^*_{[uv]_2}+x^*_{[uv]_3}\le y^*_v=0$ and thus only the $x^*$ value for $[uv]_0, [uv]_1$ can be 1. Since we have an integral solution, it follows that $x^*_{[uv]_0}+x^*_{[uv]_1}\ge 1$, since $e\in E^*$. Note that $c_{e}\ge c_{eu}$. Then $\Sigma_{i=0}^3 c_{[uv]_i}x^*_{[uv]_i} = c_ex^*_{[uv]_0}+c_{eu}x^*_{[uv]_1} \ge c_{eu}(x^*_{[uv]_0}+x^*_{[uv]_1})\ge c_{eu}=c^{Y^*}(e)$. Note that a similar argument can be made for the case when $u\notin Y^*, v\in Y^*$.
	
	Lastly, assume both $u, v\in Y^*$. Then $c^{Y^*}(e)=c_{euv}$. Note that $c_e, c_{eu}, c_{ev}\ge c_{euv}$. Then, $\Sigma_{i=0}^3 c([uv]_i) x^*_{[uv]_i} \ge \Sigma_{i=0}^3 c_{euv} x^*_{[uv]_i}\ge c_{euv}$. The last inequality is due to the fact that there exists $i=0, 1, 2, 3$ such that $[uv]_i\in F^*$. This completes the proof of claim and thus also our lemma.

\end{proof}

%% file: nvldp.tex
\section{NVDP with $k$ Levels (NVLDP)}
\label{sec:nvldp}

In this section, we prove the following theorem.

\begin{theorem}
	There exists a polytime algorithm that provides a $(4k, 4k)$-approximation to NVLDP.
\end{theorem}

The strategy is similar to that for NVDP. We first create a new graph and IP model that solves NVLDP. Then we study the solution to the LP relaxation of the problem. We similarly create an auxiliary graph $H$ and transforms it slightly so that the ball-growing method can be applied on $H$. Then by analyzing a random algorithm and use Markov's Inequality, we can show the existence of a good approximation. Then we simply greedily find such solution. The major difference is we need to be more careful when working with $H$.  In NVDP, we first turn each arc into a path of length four to create $H$ in order to better represent the four different costs associated with each arc. Then, we split the path further into two parts, aided and unaided and kept only one of the parts before running the ball-growing algorithm. If we were to repeat similarly here, we would turn each arc into a path of length $(k+1)^2$ (since we have $(k+1)^2$ many costs per arc) then we need to split it properly in order to run the ball-growing algorithm. However, this splitting step is less obvious than before. The naive way is to keep only the longest arc in the path but that blows up the approximation to a $O(k^2)$ factor. Thus, further analysis is needed in order to show (See Lemma~\ref{sparse-lemma}) that in an optimal solution to the LP, most arcs has length 0 and the path only has $O(k)$ many non-zero length arcs and thus keeping the longest arc is not too detrimental.

\paragraph{LP Model for NVLDP}
Similar to NVDP, we first transform the graph $G$ by replacing every arc $uv$ with a path of length $(k+1)^2$. The path contains arcs $[uv]_{i, j}$ where $0\le i, j\le k$ with an associated cost of $c_{i, j}(uv)$. The order of these arcs in the path does not matter but for ease of notation and consistancy, we fix in lexicographic order $[uv]_{0, 0}[uv]_{0, 1}...[uv]_{0, k}[uv]_{1, 0}...[uv]_{1, k}[uv]_{2, 0}...[uv]_{k, k}$. The vertices on the path is labelled as $(uv)_0, ..., (uv)_{(k+1)^2}$ and thus arc $[uv]_{i, j}=(uv)_{i(k+1)+j}(uv)_{i(k+1)+j+1}$. We similarly introduce cut variable $x$ for every arc in $H$ to represent whether it is in the final cut. We also introduce $k$ downgrading variables $y^i$ for each original vertex $v$ to represent whether it is downgraded to level $i$. Potential variables $d$ are also needed for every vertex in $H$. Then, we can obtain the following LP (LDLP):

	\begin{align}
	\qquad \min \qquad &\sum_{[uv]_{i, j}\in A(H)}c_{i, j}(uv)x_{[uv]_{i, j}}\notag \\
	\text{s.t.} \qquad & \sum_{i=1}^k \sum_{(vv)\in V(H)}r_i(v)y^i_v\le b   \label{LDLPb}\\
	&d_{(uv)_{i(k+1)+j+1}}\le d_{(uv)_{i(k+1)+j}} +x_{[uv]_{i, j}} \ \   \forall [uv]_{i, j}\in A(H) \label{LDLPtri}\\
	&y^i_v\ge \sum_{j=0}^kx_{[uv]_{j, i}} +\sum_{j'=0}^k x_{[vw]_{i, j'}} \ \  \forall 1\le i\le k,\ \forall \text{ paths } uvw \in G \label{LDLPint}\\
	&d_{(ss)}=0, d_{(tt)}= 1 , y^0_s=0,  y^0_t=0	\notag
	\end{align}

The intuition behind these constraints are similar to those for DLP. Constraint \ref{LDLPb} bounds the total amount of budget for downgrading. Constraint \ref{LDLPtri} is simply the triangle-inequality needed for every arc in $H$. Constraint \ref{LDLPint} relates the downgrading variable $y^i_v$ to its associated arcs. The idea is if $y^i_v=1$, then we are paying to downgrade $v$ to level $i$ and thus the cost of its incident arcs $uv, vw$ should be $c_{j, i}(uv), c_{i, j'}(vw)$ respectively (subject to how $u, w$ are downgraded). Thus, $y^i_v$ is a natural upperbound for all arcs involving $v$ at level $i$. With similar arguments, we can strengthen this to upperbound the sum of all such variables of a single $(uu)(vv)(ww)$-path, giving us constraint \ref{LDLPint}. Note that $H$ has a lot more variables than before. The number of vertices and arcs are now on the order of $(k+1)^2|A(G)|$ which is still polynomial. The number of $x, y, d$ variables are on the order of $(k+1)^2n^2, kn, (k+1)^2n^2$ respectively. Constraint \ref{LDLPb}, \ref{LDLPtri} are thus still polynomially many and for every vertex $v$, we have at most $kn^2$ many Constraint \ref{LDLPint}. Thus LDLP is still solvable in poly-time.

\paragraph{Analyzing an Optimal Solution of LDLP}
Let $(x^*, y^*, d^*)$ be an optimal solution to LDLP. We first prove the following lemma:

\begin{lemma}
\label{sparse-lemma}
	There exists an optimal solution such that for any $uv\in A(G)$, there are at most $2k+1$ non-zero values of $x_{[uv]_{i, j}}$.
\end{lemma}

\begin{proof}
Fix an arc $uv\in A(G)$ and let us look at all $x$-variables associated with this arc. Imagine the $x^*$ values are presented in a $k+1$ square matrix $M$ where $M_{i, j}=x_{[uv]_{i, j}}^*$. Define $row_i:= \Sigma_{j=0}^k M_{i, j}, col_j:= \Sigma_{i=0}^k M_{i, j}$ to be the row and column sum of $M$ respectively. Now consider the following LP:

	\begin{align} (MLP)
	\qquad \min \qquad &\sum_{0\le i, j\le k}c_{i, j}(uv)M_{i,j}\notag \\
	\text{s.t.} \qquad &\sum_{j=0}^k M_{i, j}= row_i & \forall 0\le i\le k \notag \\
	&\sum_{i=0}^k M_{i, j}= col_j & \forall 0\le j\le k \notag \\
	& M_{i, j}\ge 0 &\forall 0\le i, j\le k \notag
	\end{align}

Note that our $x^*$ is a feasible solution here. The converse is also true where if $x'$ is a feasible solution here, it can also be transformed into a feasible solution for LDLP naturally. Constraint \ref{LDLPint} remains satisfied due to the row sum and column sum constraints in MLP. The only adjustment we have to make is to the potential variables $y_{(uv)_{i(k+1)+j}}$, but this can be easily modified according to the new $x'$ values. Since the total sum of all the $x$ variables did not change, it does not affect the overall distance from $(uu)$ to $(vv)$ thus no other constraints are violated. Furthermore, it is easy to check that if $x'$ had a better objective value than $x^*$ in MLP, it would also provide a better objective value in LDLP.

Now, let us study the rank of the constraint matrix for MLP. There are $k+1$ row sum constraints and $k+1$ column sum constraints but they are linearly dependent since the sum of all rows equals to the sum of all columns. Since they are the only non-negative constraints, by the Rank Lemma for LP's, there exists an optimal solution with at most $2k+1$ non-zero values, thus proving our claim.

\end{proof}

The above technique where we isolate to study only the entries of the matrix $M$ can actually be applied to any submatrix of $M$ as well. Suppose $N$ is a submatrix of $M$ with $N_r$ rows and $N_c$ columns, then one can create a similar LP by restricting to only looking at the variables in $N$ and minimizing its cost subject to maintaining the row sum and column sum of $N$. Then, we can obtain a similar result:

\begin{corollary}
	An optimal solution of LDLP contains at most $N_r+N_c-1$ non-zero variables amongst those related to any submatrix $N$ of $M$.
\end{corollary}
	
	We will use this fact to prove the following claim and say a bit more about a solution with $2k+1$ non-zero values.

\begin{claim}
	In an optimal solution $x^*$, if $x^*$ contains $2k+1$ non-zero variables associated to $uv$, then there exist $0\le i, j\le k$ such that the $x^*$ value for $[uv]_{i, 0}, [uv]_{0, j}$ are both non-zero and they are two distinct arcs.
\end{claim}

\begin{proof}
	If we apply the corollary to the submatrix without the first row, there are at most $2k$ non-zero variables then it follows there exists at least one non-zero variable in the first row say $M_{0, j}$. Symmetrically, there exists at least one non-zero value in the first column, say $M_{i, 0}$. Now, as long as $i, j$ are not both 0, we are done. Suppose that $M_{0, 0}$ is the only non-zero value in the first row and the first column, then apply the corollary to the submatrix without the first row and column, it contains at most $2k-1$ non-zero values, contradicting the fact that there are $2k+1$ non-zero values to begin with.
\end{proof}

Now, we will use the above claims to construct our auxiliary graph $H'$ to successfully run the ball-growing algorithm. The key is to not shrink the overall distance between $s$ and $t$ too much.

\paragraph{Constructing $H'$}
For every $uv$ arc in $G$, look at the corresponding path in $H$. If there are at most $2k$ non-zero $x^*$ variables, keep only the arc with the largest $x*$ value and contract the rest. Otherwise, if there are exactly $2k+1$ non-zero $x^*$ values, there exists two distinct arcs $[uv]_{i', 0}, [uv]_{0, j'}$ with non-zero $x^*$ values. We want to group these two arcs as one object and compare its sum $x^*_{[uv]_{i', 0}}+x^*_{[uv]_{0, j'}}$ to the individual $x^*$ values of the other arcs on this path. Once again, keep only the highest value and contract the rest. This operation reduces every long $(uu)(vv)$ path to at least $1/(2k)$ of its original length. This implies the distance from $s$ to $t$ is at least $1/(2k)$.

Now we proceed with the ball-growing algorithm to find cuts at different distances. Given a cut $F$, we downgrade $v$ to level $i$ if $(vv)$ is incident to some arc in $F$ and $i=\max\{i': [uv]_{j, i'}, [vw]_{i', j'}\in F \text{ for some } u, w\in V(G), 0\le j, j'\le k\}$. Essentially, look at all arcs in $F$ that involves $v$, check which level these arcs need $v$ to be downgraded to and pick the highest one. This provides a function $L$ and a final cut $F$. The algorithm simply greedily checks all cuts from the ball growing algorithm and its associated function $L$ and $F$ until it finds one with the promised guarantee. Since there are polynomially many vertices in $H$, the algorithm only needs to check polynomially many cuts. Thus as long as one of these cuts provides the proper guarantee, the algoirthm can find it.

In order to show that such cut exists, we use a similar technique as NVDP by choosing a cut at random and looking at the expected interdiction cost and cut cost.

\paragraph{Proving the Existence of a Proper cut Using a Randomized Alogirthm} For any vertices $u, v\in H$, let $D(v)$ represent the distance from $s$ to $v$. Consider randomly choosing a number $\alpha$ between $0$ and $D((tt))$. Let $X_\alpha$ be the cut in $H'$ at distance $\alpha$ from $s$. In other words, $X_\alpha = \{[uv]_{i, j}\in A(H'): D((uv)_{i(k+1)+j})\le \alpha < D((uv)_{i(k+1)+j+1})\}$. Let $L_\alpha, F_\alpha$ be the associated level-downgrading function and cut for $X_\alpha$. We now analyze the expected downgrading and cut cost of $L_\alpha, F_\alpha$ respectively. Recall that $r_{L_\alpha(v)}(v)$ defines the index to which $v$ is downgraded by choosing the level cut at distance $\alpha$.

Given a vertex $v\in G$ and $0\le i\le k$, if the algorithm downgrades $v$ to level $i$ then an arc associated with $v$ at level $i$ must be in $X_\alpha$. Thus, let us examine these arcs which have the form $[uv]_{j, i}$ or $[vw]_{i, j'}$. Given an arc $[uv]_{j, i}\in A(H')$, note that due to the construction of $H'$, it is either the only arc between $(uu), (vv)$ or it has the form $[uv]_{0, i}$. In either cases, it is incident to $(vv)$ in $H'$. Similar statement can be said for arcs of the form $[vw]_{i, j'}$ and thus all arcs associated to $v$ at level $i$ in $H'$ are incident to $(vv)$ so these arcs forms a star $S^i_v$ centered at $(vv)$. Let $u', w'\in V(S^i_v)$ be vertices closest and farthest respectively from $s$. Note that an arc in $S^i_v$ is chosen only if $D(u')\le \alpha <D(w')$ thus the probability an arc in $S^i_v$ is chosen is at most $(D(w')-D(u'))/D((tt))$.The numerator is upperbounded by the sum of the $x^*$ value of the arc between $u', (vv)$ and $(vv), w'$. Note that this sum, in turn is upperbounded by $y^i_v$ due to Constraint \ref{LDLPint}. This implies the probability of downgrading $v$ to level $i$ is upperbounded by $(y^i_v)^*/D((tt))\le (y^i_v)^*(2k)$. Then,

\begin{align*}
	\sum_{v\in V(G)} \mathbb{E}[r_{L_\alpha(v)}(v)] \le & \sum_{v\in V(G)}\sum_{i=0}^k \mathbb{P}[v \text{ is downgraded to level } i]r_i(v)\\
	 \le & \sum_{v\in V(G)}\sum_{i=0}^k r_i(v)(y^i_v)^*(2k)\le 2kb
\end{align*}

where the last inequality is due to Constraint \ref{LDLPb}.

A similar result can be obtained for the expected cost of $F_\alpha$. Consider an arc $[uv]_{i, j}\in X_\alpha$. Note that $L_\alpha$ might end up downgrading $u, v$ beyond level $i, j$ respectively. This implies the cost of cutting $uv$ in the end is at most the cost $c_{i, j}(uv)$. Thus, $c^{L_\alpha}(F_\alpha)\le \Sigma_{[uv]_{i, j}\in X_\alpha} c_{i, j}(uv)$. Then, by linearity of expectations, it suffices to calculate the probability of an arc $[uv]_{i, j}\in A(H')$ to end up in $X_\alpha$. Using similar arguments as before, this happens only if $D((uv)_{i(k+1)+j})\le \alpha < D((uv)_{i(k+1)+j+1})$. Thus the probability is at most $x^*_{[uv]_{i, j}}/D((tt))\le 2kx^*_{[uv]_{i, j}}$. Then,

\begin{align*}
	c^{L_\alpha}(F_\alpha)\le & \sum_{[uv]_{i, j}\in A(H')} c_{i, j}(uv)\mathbb{P}[[uv]_{i, j}\in X_\alpha] \\
	\le & \sum_{[uv]_{i, j}\in A(H')}c_{i, j}(uv) 2kx^*_{[uv]_{i, j}} \le 2kopt^*
\end{align*}

where $opt^*$ is the objective value of LDLP. Then, by Markov's inequality, the probability that $\Sigma_{v\in V(G)} r_{L_\alpha(v)}(v) \le 4kb$ and $c^{L_\alpha}(F_\alpha)\le 4kopt^*$ are both independently at least $1/2$. Thus, there exists $\alpha$ such that $L_\alpha, F_\alpha$ provides the promised guarantee. 

%% file: hardness.tex
\section{Hardness Results}
\label{sec:hardness}

Hardness results~\cite{Wood93,Zen10DAM,chestnut2017hardness, chuzhoy2016} for interdiction problems typically involve a reduction from the Densest k subgraph problem which we define next.
\begin{definition}
	\textbf{Densest k Subgraph (DkS):}
	Given an {\bf undirected} graph $G$ and an integer $k$, find a vertex subset $Y\subset V(G)$ such that $|Y|=k$ and it maximizes the number of edges induced by $Y$ (i.e., with both ends in $Y$).
\end{definition}
DkS is a not only NP-hard but is also believed to be hard to approximate.
Under certain plausible complexity assumptions (such as it is hard to refute random 3-SAT instances~\cite{feige2002relations} or that there does not exist randomized subexponential time algorithms that solves NP~\cite{khot2006ruling}), there does not exist a PTAS for DkS.
Moreover, the current best approximation algorithm known for the problem~\cite{bhaskara2010detecting} has approximation ratio $O(n^{\frac14 + \epsilon})$ in an $n$-node graph for any $\epsilon > 0$.

A summary of the hardness results is in the table below.
\begin{center}
\begin{tabular}{|c|c|c|}
	\hline
	Problem & \multicolumn{2}{c|}{Hardness of}\\ \cline{2-3}
	& $(\alpha, 1)$-approximation & $(1, \beta)$-approximation \\ \hline
	
	WNVIP & DkS-Hard (Shown in \cite{chuzhoy2016,chestnut2017hardness}) & DkS-Hard (Theorem \ref{Whardbudget}) \\ \hline
	
	NVIP & DkS-Hard (Theorem 1.9 in \cite{chuzhoy2016}) & DkS-Hard (Theorem \ref{hardbudget})\\  \hline

	NVIP unit cost & DkS Hard (Appendix B in \cite{chuzhoy2016}) & DkS-Hard (Theorem \ref{uhard}) \\ \hline
	
\end{tabular}
\end{center}

Chestnut and Zenklusen~\cite{chestnut2017hardness} and Chuzhoy et. al \cite{chuzhoy2016} showed the following hardness of unicriterion approximation of the cut value for the network interdiction problem. Even though they proved this result for NFI, by our earlier observation, this applies to WNVIP directly.
\begin{theorem}[\cite{chestnut2017hardness}, Corollary 11]
	If there is an $(\alpha(n), 1)$-approximation for WNVIP, then there is a $2(\alpha(n^2))^2$-approximation for DkS.
\end{theorem}
We complement this to show a similar hardness of unicriterion approximation of the interdiction budget.
As mentioned in the introduction, since there is a simple reduction from undirected WNVIP to the directed version, we focus on the undirected version.
 \begin{theorem}
 	If there exists a $(1, \beta(n))$-approximation for the undirected version of WNVIP, then there is a $2(\beta(n^2))^2$-approximation for DkS.
 	\label{Whardbudget}
 \end{theorem}

\begin{proof}
Our strategy will be very similar to the one in prior work~\cite{Wood93,Zen10DAM,chestnut2017hardness}.
Given an instance of DkS, without loss of generality, we may assume $G$ is connected and $k<|V(G)|$. Consider the following auxiliary graph $H, V(H)=V_s\cup V_t\cup\{s, t\}$ where $V_s=V(G)$ corresponding to the original vertices and $V_t=\{t_e: e\in E(G)\}$ corresponding to the edges of $G$. Then, we add the following edges: $\{sv: v\in V_s\}, \{t_et: t_e\in V_t\}, \{vt_e: e=uv\in E(G)\}$. This graph $H$ is equivalent to subdividing every edge of $G$ then connecting $s$ to all original vertices and connecting $t$ to all subdivided edges. Note that $|V(H)|=|V(G)|+|E(G)|+2\le n^2, |E(H)|=|V(G)|+3|E(G)|\le 3n^2$ for large $n$.
	
	Now, consider the following interdiction and edge cost functions:

	\begin{center}
	\begin{tabular}{|c||c|c|c|c|c|c|}
		\hline
		Vertex & $s$ & \multicolumn{2}{c|}{$\in V_s$} & \multicolumn{2}{c|} {$\in V_t$} & $t$ \\ \hline
		$r(v)$ & $\infty$ & \multicolumn{2}{c|}{$1$} & \multicolumn{2}{c|} {$\infty$} & $\infty$ \\ \hline\hline
		Edges between & \multicolumn{2}{c|}{$s$, $V_s$} & \multicolumn{2}{c|} {$V_s$, $V_t$} & \multicolumn{2}{c|} {$V_t$, $t$} \\ \hline
		$c(e)$ & \multicolumn{2}{c|}{$\infty$} & \multicolumn{2}{c|} {$\infty$} & \multicolumn{2}{c|} {$1$} \\  \hline
	\end{tabular}
	\end{center}


Let $b=k$ and consider solving WNVIP on $H$.
For any $Y\subseteq V(G)$, denote $E_Y$ as the edges with both endpoints in $Y$.

\begin{claim}
	Let $(Y\subseteq V(H), F\subseteq E(H))$ be a solution to WNVIP.
Then $|E_Y| = |E(G)|- |F|$.
	\label{claimwbudget}
\end{claim}

Note that $Y\subseteq V_s=V(G)$ and $F$ are only edges between $V_t$ and $t$ due to the costs. First, we will show that $F$ is not incident to any $t_e$ where $e\in E_Y$. Let $e=uv\in E_Y$. Note that in the graph $H$, the neighbours of $t_e$ are $u, v, t$. This implies after interdicting $Y$, $t_e$ is only adjacent to $t$,
and hence need not be included in any minimal $st$ cut.
Next, we show that for every $e\notin E_Y$, $t_et\in F$. Suppose $e\notin E_Y$. Then, it follows $e$ is incident to some vertex $u\notin Y$. Note that $sut_et$ is a path in $H\backslash Y$. Since the cost of $su, ut_e$ is too expensive, it follows that $t_et\in F$. Since the number edges between $V_t$ and $t$ is exactly $|E(G)|$, the claim follows.

Let $c^*$ be the cost of the cut in an optimal solution to WNVIP and $l^*$ is the number of edges in a densest $k$-subgraph of $G$. It follows from the above claim that $l^*=|E(G)|-c^*$. Suppose the approximation scheme produced an interdiction set $V'\subseteq V_s$ with a final cut $E'$ and a cost of $c'$ where $|V'|=\Sigma_{v\in V'} r(v) \le \beta(|V(H)|)b=\beta(n^2)k$ and $c'\le c^*$. Then, it follows from the claim that $|E_{V'}| = |E(G)|-c' \ge |E(G)|-c^* =l^*$.
Now we can apply the following lemma from \cite{chestnut2017hardness}.
\begin{lemma}
	Given a graph $H$ with $n$ nodes and $m$ edges, there exists a deterministic polynomial algorithm that produces a subgraph on $k$ vertices with at least $\frac{k(k-1)}{n(n-1)} m$ edges for any $k \leq n$.
	\label{findksub}
\end{lemma}


By applying Lemma \ref{findksub} on the subgraph induced by $V'$, there exists a $k$-vertex subgraph with at least $\frac{k(k-1)}{\beta(n^2)k(\beta(n^2)k-1)}l^*\ge \frac{l^*}{2(\beta(n^2))^2}$ edges. Then, our theorem follows.


\end{proof}


Our goal in this section is to build on the proof of Theorem \ref{Whardbudget} to show hardness of NVIP with unit costs. We will do this in two steps: first we consider unitary interdiction costs with general edge cut costs. Then we also transform the edges so they have unit costs. This method allows us to show hardness for one of the unicriterion approximations.

\begin{theorem}
	If there exists a $(1, \beta(n))$-approximation for NVIP, then there is a $4(\beta(n^2))^2$-approximation for DkS.
	\label{hardbudget}
\end{theorem}

Note that this is sufficient to show a $(1, \beta)$-approximation is hard to obtain for NVIP.

\begin{proof}
	
	We once again consider the same auxiliary graph $H$ as in the proof of Theorem~\ref{Whardbudget}. However, we consider the case of unit interdiction costs.
	
\begin{center}
	\begin{tabular}{|c||c|c|c|c|c|c|}
		\hline
		Vertex & $s$ & \multicolumn{2}{c|}{$\in V_s$} & \multicolumn{2}{c|} {$\in V_t$} & $t$ \\ \hline
		$r(v)$ & $\infty$ & \multicolumn{2}{c|}{$1$} & \multicolumn{2}{c|} {$1$} & $\infty$ \\ \hline\hline
		Edges between & \multicolumn{2}{c|}{$s$, $V_s$} & \multicolumn{2}{c|} {$V_s$, $V_t$} & \multicolumn{2}{c|} {$V_t$, $t$} \\ \hline
		$c(e)$ & \multicolumn{2}{c|}{$\infty$} & \multicolumn{2}{c|} {$\infty$} & \multicolumn{2}{c|} {$1$} \\  \hline
	\end{tabular}
\end{center}

	Our budget $b=k$. Note that we still forbid the interdiction of $s$ and $t$. This is a natural condition to impose on NVIP. The main difference here compared to the proof of Theorem \ref{Whardbudget} is that we are now allowed to interdict vertices in $V_t$.
	
	Suppose $Y\subseteq V(H), F\subseteq E(H)$ is a $(1, \beta)$-approximation solution for NVIP on $H$. Then, consider an alternate interdiction set where we still interdict every vertex in $Y\cap V_s$ but instead of interdicting any $t_e\in Y\cap V_t$ corresponding to an edge $e=uv$, we instead interdict $s_u$ and $s_v$. Formally, let $\bar{Y}=(Y\cap V_s)\cup \{s_u, s_v: t_e\in Y\cap V_t, e=uv\}$. Note that $|\bar{Y}|\le 2|Y|$. Thus, $\bar{Y}, F$ is a $(1, 2\beta)$-approximation for NVIP on $H$. Applying Claim \ref{claimwbudget}, we can obtain a $2\beta(n^2)k$ vertex subgraph with at least $l^*$ edges where $l^*$ is the optimal number of edges in DkS. Then, the result follows from Lemma \ref{findksub}.
\end{proof}

We can further modify the graph $H$ to show that we cannot obtain a good unicriterion approximation for NVIP even when we have unitary costs.
\begin{theorem}
	If there exists a $(1, \beta(n))$-approximation for NVIP with unit cut cost, then there exists a $(1, 4(\beta(2n^2))^2)$-approximation for DkS.
	\label{uhard}
\end{theorem}

\begin{proof}
	We will use the same set up as Theorem \ref{hardbudget}. Recall, given graph $G$, we construct auxiliary graph $H$ with the following cost functions:
	
\begin{center}
	\begin{tabular}{|c||c|c|c|c|c|c|}
		\hline
		Vertex & $s$ & \multicolumn{2}{c|}{$\in V_s$} & \multicolumn{2}{c|} {$\in V_t$} & $t$ \\ \hline
		$r(v)$ & $\infty$ & \multicolumn{2}{c|}{$1$} & \multicolumn{2}{c|} {$1$} & $\infty$ \\ \hline\hline
		Edges between & \multicolumn{2}{c|}{$s$, $V_s$} & \multicolumn{2}{c|} {$V_s$, $V_t$} & \multicolumn{2}{c|} {$V_t$, $t$} \\ \hline
		$c(e)$ & \multicolumn{2}{c|}{$\infty$} & \multicolumn{2}{c|} {$\infty$} & \multicolumn{2}{c|} {$1$} \\  \hline
	\end{tabular}
\end{center}

	There are two edge costs we need to modify, those between $s, V_s$ and those between $V_s, V_t$. To take care of those incident to $s$, consider adding a large clique $S=K_{n^2}$ between $s$ and $V_s$. Then, add every possible edge between $s, S$ and similarly between $S$ and $V_s$. Note that even after interdicting $b<n$ nodes in $S$, any cut through $S$ involves $n^2 > |E(G)|$ many edges and thus any minimum solution will not interdict any vertices in $S$ nor cut any edges in $S$.
	
	For edges between $V_s, V_t$, we simply set their cost to $1$ and claim that a $st$-minimum cut would not use any of those edges either. Let $F$ be a minimum cut after interdicting some vertices $Y$. Suppose there exists $v\in V_s, t_e\in V_t$ such that $vt_e\in F$. Then, removing edge $vt_e$ and adding $t_et$ to $F$ is still a $st$-cut. Therefore, we can assume that any min-cut after interdicting any set $Y$ only consists of edges incident to $t$. Then, applying the same techniques in Theorem \ref{hardbudget}, our result follows.
	
\end{proof}

Note that it is not known if it is also hard to obtain an $(\alpha, 1)$-approximation for NVIP and NVIP with unit cost.

%% file: vip.tex
\section{Simple Approximations for Interdiction Problems}
\label{sec:vip}

\subsection{A Simple Bicriteria Algorithm for WNVIP}

We first provide a polynomial time algorithm that finds a $(1+1/\epsilon, 1)$ or $(1, 1+\epsilon)$-approximation for any $\epsilon >0$ for WNVIP in digraphs using a more direct proof than the earlier methods~\cite{burch03,chestnut2016interdicting}.
This proof was also given by Chuzhoy et al.~\cite{chuzhoy2016} in their study of $k$-route cuts, but we reproduce it since we build on it later. A priori, the algorithm will not be able to tell which guarantee it provides. Note that if $\epsilon=1$, this easily provides a $(2, 2)$ approximation.

\begin{theorem}
\label{th:nvip}
For any $\epsilon>0$, there exists a polynomial time algorithm that provides either a $(1+1/\epsilon, 1)$ or a $(1, 1+\epsilon)$-approximation guarantee for WNVIP in directed graphs.
\end{theorem}

\begin{proof}
Given $\epsilon>0$, let $L$ be our guess for the value of $\lambda_{st}(G\backslash X^*)$ for the optimal interdiction set $X^*$ that obeys $r(X^*) \leq b$. Since we can detect whether $L = 0$ (using weighted node connectivity computation) and output an interdiction set achieving this, we assume that $L$ is nonzero for the rest of the proof. Also, note that $L$ is bounded above by the weight of a minimum directed $st$-cut of $G$. Consider an auxiliary digraph $G'$ and a weight function $w:A(G')\to \mathbb{R}_+$ where we subdivide every arc $e$ with a vertex $v_e$. We assign weight $\epsilon c_eb/L$ to every new vertex $v_e$. Every original vertex $v$ gets weight $r(v)$. If our guess $L$ is correct, then $G'$ contains a $st$-separating vertex cut of weight $(1+\epsilon)b$ (by using the nodes in an optimal interdicting set and the subdivided nodes of the corresponding min $st$-cut). Now, consider a minimum weighted vertex cut $X'$ of $G'$ that separates all directed paths from $s$ to $t$. Let $X= X'\cap V(G)$ and $Y=X'\backslash X$ be those that corresponds to subdivided arcs. Note that $(1+\epsilon)b\ge w(X')=r(X)+\Sigma_{e\in Y}\epsilon c_eb/L$. In particular, $r(X)\le (1+\epsilon)b$ and $\Sigma_{v_e\in Y} c_e\le (1+1/\epsilon)L$. Furthermore, if $r(X) > b$, then it follows that $\Sigma_{e\in Y}c_e<L$. This implies if $L$ is the correct guess, then $(X, Y)$ is either a $(1+1/\epsilon, 1)$ or a $(1, 1+\epsilon)$-approximation\footnote{We can simply try all possible $L$ values using binary search to find the smallest value for which the condition holds. To make this search polynomial time, we can simply search over multiplicative powers of $(1 + \epsilon')$ for some small $\epsilon' > 0$ starting from 1 and up to the maximum possible cost of a cut using $\log_{1 + \epsilon'} {mc_{max}}$ trials where $m$ is the number of arcs and $c_{max}$ is the largest arc cost.}.
\end{proof}


\subsection{Interdicting Multiway Cuts}

\begin{problem}
	\textbf{Weighted Multiway Cut Vertex Interdiction Problem}
	
	\textbf{(WMWIP)} Let $G$ be an directed graph, $S=\{s_1, ..., s_k\}\subseteq V(G)$ and every arc $e$ has a non-negative weight $c_e$ and every vertex $v$ has a non-negative interdiction cost $r(v)$. Given a budget $b$, find vertices $X$ of interdiction cost at most $b$ to interdict that minimizes the cost of separating the vertices in $S$. In other words, find $X\subseteq V(G), F\subseteq E(G)$ such that there is no path from any vertex $s_i\in S$ to any other $s_j \in S$ in $G$ after deleting X and $F$. Furthermore, $r(X)\le b$ and $\Sigma_{e\in F}c_e$ is minimized. We will assume that the demand vertices $S$ cannot be interdicted or equivalently that $r(s) = \infty$ for $s \in S$.
\end{problem}

\paragraph{Related Work.} The multiway cut problem is another natural extension of the min $st$-cut problem, and involves finding a minimum set of edges to delete in order to separate $k$ given terminals from each other. For the problem in undirected graphs, using a geometric relaxation, Calinescu, Karloff and Rabani~\cite{cualinescu2000improved} achieved an approximation factor 3/2. Sharma and Vondrak~\cite{Sharma2014MultiwayCP} gave the current best approximation factor of 1.2965. In the vertex version of the Multiway Cut Problem,  instead of deleting edges, we delete vertices to separate the $k$ terminals, but the graph is still undirected.
By subdividing edges and introducing new nodes with costs, it is easy to see that the vertex version generalizes the edge version.
Garg et al.~\cite{garg1994multiway} gave a 2-approximation for this vertex Multiway Cut in undirected graphs by showing half-integrality of the optimal solution of a natural LP relaxation.
Finally, the directed version of the problem involves deleting arcs so that there are no directed paths between any pair of terminals.
It is easy to see that the vertex version of this directed problem can be reduced back to the arc version by the usual splitting of nodes into two copies: one for supporting incoming and the other for outgoing arcs with the node cost assigned to the arc from the in- to the out-copy.
Naor and Zosin~\cite{naor2012} gave the first 2-approximation algorithm for the directed multiway cut problem using a sophisticated LP formulation and rounding.
Chekuri and Madan~\cite{chekuri2016simple} later gave a simple ball growing based LP rounding algorithm for the problem that also gives a $2(1 - \frac1k)$ approximation where $k$ is the number of terminals. We can adapt either algorithm to give our results for the vertex interdiction variant of multiway cuts in directed graphs.

Using similar techniques as for WNVIP, we prove the following for the vertex interdiction version of the directed multiway cut problem.

\begin{theorem}
	For any $\epsilon>0$, there exists a polynomial time algorithm that gives a $(4, 4(1+\epsilon))$-approximation to Weighted Multiway Cut Vertex Interdiction Problem (WMWIP) in directed graphs.
	\label{th:WMWIP}
\end{theorem}

\begin{proof}
	Let $opt$ be the cost of an optimal solution to WMWIP. First delete all arcs of cost $0$ since they are free to cut. Then, apply the 2-approximation algorithm for the node-weighted version of the directed multiway cut problem on $G$ disregarding any edge costs. If $opt=0$, then the algorithm should produce a set of vertices to interdict whose cost is at most $2b$ and separates all terminals from each other.
	
	Now, assume $opt>0$ and we now make a guess of a close range bounding $opt$. First, by scaling the costs with an appropriate factor, we may assume that all arc costs are at least $1$. Fix a constant $\epsilon>0$ and let $q = \lceil\log(\sum_{uv\in A} c_{uv})/\log(1+\varepsilon)\rceil$. Define the sequence $U_i = (1+\varepsilon)^i$ for $i=0,\ldots,q$. By the choice of $q$, we have $(1+\varepsilon)^q> \sum_{uv\in A} c_{uv}$. Clearly, $U_0 \leq opt < U_q$. We improve these bounds by guessing an index $i_0$ such that $U_{i_0-1} \le opt \le U_{i_0}$. This means that the following algorithm tests $U_{i_0} = U_i$ for $i=1,\ldots,q$ and returns the best result.
	
	Consider an auxiliary graph $G'$ and a weight function $w:A(G')\to \mathbb{R}_+$ where we subdivide every arc $e$ with a vertex $v_e$. We assign weight $\frac{c_eb}{U_{i_0}}$ to every new vertex $v_e$. Every original vertex $v$ gets weight $r(v)$. All original arcs are assumed to have infinite weight.
Then $G'$ contains a node multiway cut of weight $2b$ (by using the nodes in an optimal interdicting set and the subdivided nodes of the corresponding minimum multiway cut).
Recall that the vertex version of this directed problem can be reduced back to the arc version by the usual splitting of nodes into two copies: one for supporting incoming and the other for outgoing arcs with the node cost assigned to the arc from the in- to the out-copy. Now, we use any existing 2-approximation for directed multiway cut~\cite{naor2012,chekuri2016simple} to find a 2-approximation to the minimum  weighted vertex multiway cut to get the set $X'$ of $G'$. Let $X= X'\cap V(G)$ and $Y=X'\backslash X$ be those that corresponds to subdivided arcs. Note that $4b\ge w(X')=r(X)+\Sigma_{e\in Y}\frac{c_eb}{U_{i_0}}$. In particular, $r(X)\le 4b$ and $\Sigma_{e\in Y} c_e\le 4U_{i_0} \le 4(1+\varepsilon)opt$.
\end{proof}

\subsection{Interdicting MultiCuts Using Vertices}

\begin{problem}
	\textbf{Weighted MultiCut Vertex Interdiction Problem (WMVIP)} Let $G$ be an {\bf undirected} graph where every edge $e$ has a non-negative weight $c_e$ and every node $v$ has a non-negative interdiction cost $r(v)$, and we are given demand pairs of vertices $\{s_1,t_1\}, \ldots,\{s_k,t_k\}$. Given a budget $b$, find nodes $X$ of interdiction cost at most $b$ to interdict that minimizes the cost of separating the demand pairs in the resulting graph. In other words, find $X\subseteq V(G), F\subseteq E(G)$ such that for every demand pair $\{s_i,t_i\}$, vertices $s_i$ and $t_i$ are in different connected components of $G$ after deleting the nodes in $X$ and the edges in $F$. Furthermore, $r(X)\le b$, and $\Sigma_{e\in F}c_e$ is minimized.
	
We will assume that the demand-pair vertices cannot be interdicted or equivalently that $r(v) = \infty$ for $v \in \{s_i,t_i\}$ for any demand pair $i$, and call such vertices terminals.
\end{problem}

\paragraph{Related Work.}
In the multicut problem in {\bf undirected} graphs, we are given $k$ source-sink pairs and a multicut puts every source sink pair in different connected components. The minimum multicut problem has a well-known $2 \ln k$-approximation algorithm~\cite{garg1996approximate} using an LP-rounding method. An alternate proof of this result using the ideas of Calinescu et al.~\cite{cualinescu2000improved} uses a randomized Dijkstra-like ball growing and cutting method (See e.g.,~\cite{GO2008}) that we adapt in designing our approximation algorithm for the vertex interdiction variant of this problem.
For completeness, we remark that multicuts in directed graphs are not that well approximable with the best known approximation ratio being $O(n^{\frac{11}{23}})$~\cite{agarwal2007improved} in $n$-node digraphs.

This subsection focuses on applying a similar technique as above on WMVIP to prove the following theorem.

\begin{theorem}
	For any $\epsilon>0$, there exists a randomized polynomial time algorithm that gives a $(2(1+\epsilon) \ln k, 2 (1+\epsilon)\ln k)$-approximation to Weighted MultiCut Vertex Interdiction Problem (WMVIP) in {\bf undirected graphs}, where $k$ is the number of terminal pairs.
	\label{th:WMIP}
\end{theorem}

Our strategy is to formulate and solve a linear programming relaxation for the problem by adapting the multicut formulation by incorporating node interdiction variables. We then employ a ball-growing based rounding technique used for deriving a logarithmic approximation algorithm by first transferring the LP values on the nodes to all its adjacent edges.
We  then observe that this transformation does not degrade the quality of the final approximation more than the claimed amount.

Consider the following linear programming relaxation for WMVIP:
\begin{align}
\max \qquad &\sum_{e\in E(G)}c_ex_e \notag\\
\text{s.t.} \qquad &\sum_{e\in E(P)}x_e+\sum_{v\in V(P)}y_v\ge 1 &  \forall s_it_i \ \text{path } P \forall i \notag \\
&\sum_{v\in V(G)}r(v)y_v\le b &  \notag \\
&x_e, y_v \ge 0& \forall e\in E(G), v\in V(G)\backslash S\notag\\
& y_v=0 & \forall \text{terminals }v\notag
\end{align}
Note that an optimal solution to the multicut interdiction problem (WMVIP) is a feasible solution to the above LP. Let $x^*, y^*$ be an optimal solution to this LP and let $opt^*$ be the optimal value. Note that if we define the distance between any two points $u, v$ as $\min_{uv-\textbf{path}\ P} \Sigma_{e\in E(P)}x_e+\Sigma_{w\in V(P)\backslash \{u, v\}}y_w$, then every $s_i, t_i$ pair is at least 1 unit apart from each other. We will now construct an auxiliary graph $G'$ that transfer all weights on vertices to edges and preserves this distance.

To construct the auxiliary graph $G'$, for every edge $e=uv$, subdivide it twice with vertices $u_v, v_u$ and give weights $y_u/2, x_e, y_v/2$ to $uu_v, u_vv_u, v_uv$ respectively. This transformation can be viewed as replacing every vertex $v$ with a star with center $v$ and leaves $v_u$ for every $u\in N(v)$. Note that the length of any path is clearly preserved. Thus every $s_i, s_j$ pair is still distance 1 apart. Let $d(u, v)$ denote the distance between two vertices $u, v$. Let edges of the form $uu_v,v_uv$ be called \textit{vertical} edges since they relate to the original vertex weights $y_u,y_v$. Denote all other edges of the form $u_vv_u$ as \textit{true} edges.

Consider the following Dijkstra-like Ball Growing Algorithm for multicut~\cite{GO2008} inspired by~\cite{cualinescu2000improved}: Choose a random permutation of the demand pairs and reindex the pairs according to this random order. Next, randomly choose a number $r$ between $(0, 1)$. In increasing order $i$ of the terminal pairs in this random ordering, draw a ball of radius $r$ centered at $s_i$. Then, cut all edges at the boundary of the ball around $s_i$ that are not cut or entirely contained inside the ball around $s_j$ for some $j<i$. In other words, an edge $e=uv$ is chosen as part of the cut if and only if there exists $1\le i\le k$ such that either $d(u, s_i)\le r < d(v, s_i)$ or $d(v, s_i)\le r<d(u, s_i)$ and $d(u, s_j), d(v, s_j)> r$ for all $j< i$ in the random permutation order of the terminal pairs. Then given $r$ and the random order of the demand pairs, let $F$ be the set of edges in $G$ that corresponds to true edges chosen in the cut and let $X$ be the set of vertices of $G$ that corresponds to vertical edges in the cut.

Note that the resulting $X, F$ does provide an integral solution to the WMVIP problem. Now, we calculate the expected cost of $F$ and $X$.

\begin{lemma}
	The probability that a true edge $u_vv_u$ is chosen in the cut is at most $x_e\ln k$.
\end{lemma}

\begin{proof}
	Let $e=u_vv_u$ be a true edge. For a demand pair $(s, t)$, denote the distance of the edge $e$ to the terminal $s$ as $\min\{d(u_v, s), d(v_u, s)\}$. Then, we rank the demand pairs $(s'_i, t'_i)$ based on the distance of $e$ to the terminals $s_i$. Suppose that $s'_i$ is the $i$-th closest terminal to $e$.
	
	Without loss of generality, assume $d(u_v, s'_i)\le d(v_u, s'_i)$. Then, $e$ is on the boundary of $s'_i$ with radius $r$ if and only if $d(u_v, s'_i)< r< d(v_u, s'_i)$ which happens with probability $d(v_u, s'_i)-d(u_v, s'_i)\le x_e$. Consider $j<i$ so $e$ is closer to $s'_j$ than $s'_i$. Suppose $e$ is on the boundary of $s'_i$ and thus $r> d(s'_i, u_v)$. Then $r$ is larger than the distance to $s'_j$. This implies $e$ is either inside or on the boundary of the ball around $s'_j$ with radius $r$. Then $e$ cannot be chosen as part of the cut by $s'_i$. Thus $e$ is chosen in the cut by $s'_i$ only if $s'_i$ is chosen in the random ordering after $s'_j$ for all $j<i$. Then, the probability $e$ is chosen in the cut by $s'_i$ is at most $x_e\times \frac1i$. Then, the total probability of $e$ being chosen in the cut is at most $\Sigma_{i=1}^k \frac{x_e}{i} \leq x_e\ln k$.
\end{proof}

A similar result can be derived about interdicting a vertex.
\begin{lemma}
	The probability that $v$ is interdicted is at most $y_v\ln k$.
\end{lemma}

\begin{proof}
	Define $\delta_v=\min_{u\in N(v)}\{d(v_u, s)\}$ as the distance between $v$ and a terminal $s$. Once again, we can rank the terminals based on their distance to $v$ and thus assume that $s'_i$ is the $i$-th closest terminal to $v$. Note that the maximum distance between any two vertices incident to $v$ is $y_v$ due to the path via $v$.
	
	For similar reasons as before, the probability of a vertical edge of $v$ to lie on the boundary of a ball centered around $s'_i$ is $y_v$. However, it is chosen in the cut due to $s'_i$ only if $s'_i$ appears in the ordering after $s'_j$ for all$j<i$. Then, by similar reasoning, the probability $v$ is interdicted is at most $y_v\ln k$ .
\end{proof}

Now we can prove our main theorem.

\begin{proof} (Theorem \ref{th:WMIP})
	Note that $\mathbb{E}[\Sigma_{e\in F}c_e] = \Sigma_{e\in E(G)} c_ePr[e\in F]$ is at most $\Sigma_{e\in E(G)}c_ex_e\ln k=opt^*\ln k$. Similarly, $E[\Sigma_{v\in X} r(v)]=\Sigma_{v\in V(G)} r(v)Pr(v\in X)$ is at most $\Sigma_{v\in V(G)}r(v)y\ln k_v=b\ln k$. Then by Markov's inequality, the probability that the corresponding $F, X$ satisfy $\Sigma_{e\in F} c_e\le 2(1+\epsilon)opt^*\ln k, \Sigma_{v\in X} r(v)\le 2(1+\epsilon)b\ln k$ is $1-\frac{2}{2(1+\epsilon)}=\frac{\epsilon}{1+\epsilon}$. Thus, we can find a desirable cut in polynomial time.
\end{proof}
